\documentclass[11pt,a4paper]{article}

\usepackage{newunicodechar}

\newunicodechar{┌}{+}

\newunicodechar{┐}{+}

\newunicodechar{└}{+}

\newunicodechar{┘}{+}

\newunicodechar{─}{-}

\newunicodechar{│}{|}

\usepackage[a4paper, top=28mm, bottom=28mm,
            left=30mm, right=30mm]{geometry}
\usepackage{setspace}
\usepackage{microtype}
\usepackage{amsmath, amssymb, amsthm, bm}
\usepackage{mathtools}

\usepackage{booktabs}
\usepackage{array}
\usepackage{longtable}
\usepackage{xcolor}
\usepackage{graphicx}
\usepackage{listings}

\lstdefinestyle{scriptlog}{
  basicstyle=\ttfamily\scriptsize,
  breaklines=true,
  breakatwhitespace=false,
  frame=single, framerule=0.4pt, rulecolor=\color{gray!40},
  numbers=none,
  showstringspaces=false,
  tabsize=2,
  keepspaces=true,
  columns=flexible,
  inputencoding=utf8,
  extendedchars=true,
}

\lstdefinestyle{python}{
  language=Python,
  basicstyle=\ttfamily\small,
  keywordstyle=\color{blue!70!black}\bfseries,
  commentstyle=\color{green!50!black}\itshape,
  stringstyle=\color{red!60!black},
  numbers=left, numberstyle=\tiny\color{gray},
  numbersep=8pt,
  frame=single, framerule=0.4pt, rulecolor=\color{gray!40},
  breaklines=true, showstringspaces=false,
  tabsize=4, extendedchars=true, inputencoding=utf8
}

\usepackage{fancyhdr}
\fancypagestyle{plain}{%
  \fancyhf{}
  
  \fancyfoot[C]{\small\thepage}}

\usepackage[colorlinks=true, linkcolor=blue!60!black,
            citecolor=green!40!black, urlcolor=blue!60!black]{hyperref}

\theoremstyle{plain}
\newtheorem{theorem}{Theorem}[section]

\newtheorem{corollary}[theorem]{Corollary}
\theoremstyle{definition}
\newtheorem{definition}[theorem]{Definition}
\newtheorem{algorithm_env}[theorem]{Algorithm}
\newtheorem{remark}[theorem]{Remark}

\newcommand{\balpha}{\boldsymbol{\alpha}}
\newcommand{\bx}{\mathbf{x}}
\newcommand{\by}{\mathbf{y}}
\newcommand{\ba}{\mathbf{a}}
\newcommand{\bv}{\mathbf{v}}

\newcommand{\bdelta}{\boldsymbol{\delta}}
\newcommand{\bzero}{\mathbf{0}}
\newcommand{\bU}{\mathbf{U}}
\newcommand{\bPhi}{\boldsymbol{\Phi}}
\newcommand{\bSigma}{\boldsymbol{\Sigma}}
\newcommand{\R}{\mathbb{R}}
\newcommand{\Z}{\mathbb{Z}}
\newcommand{\kker}{\operatorname{ker}}
\newcommand{\rank}{\operatorname{rank}}

\newcommand{\col}{\operatorname{col}}

\newcommand{\Mach}{\mathrm{Ma}}
\newcommand{\fl}{\varphi}        
\newcommand{\hc}{\psi}           

\newcommand{\keywords}[1]{%
  \smallskip\noindent\textbf{Keywords:}\enspace#1\par\smallskip}

\title{\textbf{The Algebra of Units: From Buckingham's $\Pi$ Theorem
        to Latent-Variable Learning}\\[6pt]
       \large Discovering Dimensionless Groups from Data Without Prior Physics}
\author{M.~Valorani\footnote{Corresponding Author. 
\textsl{E-mail address}: mauro.valorani@uniroma1.it} \\
\textsl{Department of Mechanical and Aerospace Engineering} \\ 
\textsl{Sapienza University of Rome}
}
\date{}

\begin{document}

\maketitle

\begin{abstract}
When engineers study how a machine or a physical phenomenon behaves,
they often measure many different quantities --- speed, pressure,
temperature, length --- all in different physical units.
A classical idea, known as the Buckingham $\Pi$ theorem (1914), says
that you can always combine these measured quantities into a smaller
set of \emph{dimensionless numbers} --- ratios in which all units
cancel --- and that the physics of the system depends only on those
ratios, not on the original measurements.
Famous examples are the Reynolds number in fluid mechanics, the Mach
number in aerodynamics, and the efficiency and flow coefficient used
to characterise fans and compressors.

Finding the right dimensionless numbers has traditionally required
expert knowledge: an engineer had to know, from first principles,
which physical quantities matter and how they are related.

This paper shows that the same dimensionless numbers can be
\emph{discovered automatically from data}, without any prior
knowledge of the underlying physics.
The key idea is simple: if you take the logarithm of your
measurements and then look at how they vary when you rescale
the experiment (for example, by running the same machine at a
different size or a different speed), the data falls onto a
low-dimensional ``manifold'' whose shape is entirely determined by
the dimensionless groups.
A standard linear algebra tool --- singular value decomposition,
the same technique used in image compression and recommendation
systems --- can identify that manifold directly.
A final step searches all integer-exponent combinations within that
manifold and retains only those built around the machine's own
characteristic scales; this \emph{repeating-variable filter} cleanly
recovers the familiar named groups (flow coefficient, head
coefficient, Mach number) and distinguishes them from equivalent
but less interpretable alternatives.

The method is demonstrated on a synthetic compressor dataset with
16\,000 measurements.
It recovers the correct dimensionless groups to numerical precision
and fits the compressor performance map with an error below
$0.01\%$, starting from raw dimensional data and no physics input.

The broader message is that two fields --- classical dimensional
analysis and modern data-driven machine learning --- are not as
separate as they appear.
They share the same underlying algebra, and making that connection
explicit opens new ways to build physical models that are both
interpretable and data-efficient.
\end{abstract}


\tableofcontents
\bigskip\hrule\bigskip

\section{Introduction}
\label{sec:intro}

The Buckingham $\Pi$ theorem \cite{Buckingham1914} and Proper
Orthogonal Decomposition (POD) are usually presented in entirely
separate courses: the former in classical mechanics and fluid dynamics,
the latter in computational methods and data science.
This paper argues that the two are not merely analogous but
mathematically related, and that recognising this relationship
opens a productive bridge between classical dimensional analysis
and modern machine learning.

Both methods are, at their core, \emph{dimensionality reductions}.
The Buckingham theorem reduces $k$ dimensional variables to $k-n$
dimensionless $\Pi$ groups, discarding the $n$ degrees of freedom
associated with the choice of base units.
POD reduces an $N$-dimensional field representation to $r \ll N$
modal coordinates ($r$ denoting the number of retained modes).
The question we address is: when the data analysed by POD is physical,
what is the relationship between the POD latent coordinates and the
$\Pi$ groups?

The answer involves three ingredients.
\begin{enumerate}
  \item \textbf{The logarithmic bridge.}  A logarithmic change of
    variables linearises the multiplicative structure of dimensional
    monomials.  $\Pi$ groups become inner products of log-data with
    null-space vectors, and experimental data lies on a
    $(k{-}n)$-dimensional affine subspace of log-space whose direction
    is $\kker(\balpha^T)$.

  \item \textbf{Gauge variation and within-cluster SVD.}
    The $n$ gauge directions (associated with the choice of unit
    scales) and the $k-n$ physical directions (the $\Pi$-group
    subspace) can be separated by generating data with controlled
    \emph{gauge variation}: replicating each operating condition at
    multiple physical scales.  A within-cluster SVD then recovers
    $\kker(\balpha^T)$ to machine precision, with an exact gap in the
    singular-value spectrum at position $n$.

  \item \textbf{Integer lattice search.}
    The named $\Pi$ groups --- flow coefficient $\fl$, head
    coefficient $\hc$, Mach number $\Mach$ --- are the generators
    of the integer lattice $\kker(\balpha^T)\cap\Z^k$ that
    correspond to the standard Buckingham basis formed with $N$ and
    $D$ as repeating variables.  They cannot be recovered by
    orthogonal rotation of the SVD frame because the natural groups
    are generically non-orthogonal as exponent vectors.
    A brute-force integer lattice search, feasible for engineering
    unit systems with maximum exponent magnitude $m_{\max}\le 4$
    (see Algorithm~\ref{alg:lattice}), enumerates the full catalogue
    of primitive lattice vectors from which the named groups are identified.
\end{enumerate}

The theoretical framework is illustrated on a turbomachinery test case
designed to answer three questions simultaneously:
(i) Can $\kker(\balpha^T)$ be recovered from dimensional data without
prior knowledge of the unit structure?
(ii) How many $\Pi$ groups are needed to parameterise a given
performance map --- and can this number be detected from data?
(iii) Once the dimensionless coordinates are identified, can the map
be reconstructed with engineering accuracy, and can the named Pi
groups be re-discovered?

\subsection{Organisation}

The paper is organised as follows.
Sections~\ref{sec:buckingham}--\ref{sec:lattice} develop the
theoretical framework (null-space formulation, logarithmic bridge,
SVD recovery, integer lattice search).
Section~\ref{sec:paramsol} connects the framework
to POD in solution space.
Section~\ref{sec:testcase} presents the turbomachinery test case
end-to-end.
Section~\ref{sec:discussion} discusses the workflow, practical
implications, connection to incomplete similarity, and open research
directions.

\section{The Buckingham $\Pi$ Theorem as a Null-Space Problem}
\label{sec:buckingham}

The material in this section is classical; standard references include
Buckingham~\cite{Buckingham1914} and Bridgman~\cite{Bridgman1922}.

Let $q_1, \ldots, q_k$ be physical variables, each with a dimensional
signature expressed in $n$ base units $u_1, \ldots, u_n$:
\[
  [q_j] = u_1^{\alpha_{j1}} u_2^{\alpha_{j2}} \cdots u_n^{\alpha_{jn}},
  \qquad \alpha_{ji} \in \Z.
\]
Assembling these exponents into the \emph{dimensional matrix}
$\balpha \in \Z^{k \times n}$, a monomial
$\Pi = \prod_{j=1}^k q_j^{x_j}$ is dimensionless if and only if
\begin{equation}
  \label{eq:nullspace}
  \balpha^T \bx = \bzero, \qquad \bx = (x_1, \ldots, x_k)^T \in \R^k.
\end{equation}
The set of all solutions is the null space $\kker(\balpha^T)$.
By the rank-nullity theorem applied to $\balpha^T : \R^k \to \R^n$,
\[
  \dim\kker(\balpha^T) = k - \rank(\balpha^T) = k - n
  \quad \text{(when } \rank(\balpha^T) = n\text{)}.
\]
This is the Buckingham count: there are exactly $k - n$ independent
dimensionless groups, each corresponding to a basis vector of
$\kker(\balpha^T)$.
The basis is not unique: any non-singular linear combination of basis
vectors yields an equally valid set of $\Pi$ groups.

\section{The Logarithmic Bridge}
\label{sec:log}

\subsection{Linearising monomials}

The logarithm transforms the multiplicative structure of dimensional
monomials into a linear one.
Define log-variables $\by = (\log q_1, \ldots, \log q_k)^T \in \R^k$.
A $\Pi$ group with exponent vector $\ba \in \kker(\balpha^T)$ becomes
\begin{equation}
  \label{eq:logpi}
  \log \Pi = \sum_{j=1}^k a_j \log q_j = \ba \cdot \by.
\end{equation}
In log-space, a $\Pi$ group is an inner product with a null-space vector.
The $k-n$ independent $\Pi$ groups therefore span the
$(k{-}n)$-dimensional subspace $\mathcal{P} = \kker(\balpha^T)\subset\R^k$.

\subsection{The physical manifold}

Consider a set of $m$ experiments measuring the $k$ variables.
Let $X\in\R^{m\times k}$ be the data matrix whose $(i,j)$ entry is the
$j$-th variable recorded in experiment $i$.
Assemble the log-data matrix $Y = \log X \in \R^{m\times k}$
(entry-wise logarithm).  Under the axiom of dimensional homogeneity, the rows of
$Y$ are constrained to a $(k-n)$-dimensional affine subspace.

\begin{theorem}[Physical manifold in log-space]
  \label{thm:manifold}
  Suppose the $m$ experiments are governed by a physically homogeneous
  relation $F(\Pi_1, \ldots, \Pi_{k-n}) = 0$, and that the experiments
  span the full range of the independent $\Pi$ groups.
  Then the rows of $Y$ lie in a $(k-n)$-dimensional affine subspace of
  $\R^k$ whose direction space is $\kker(\balpha^T)$.
\end{theorem}

\begin{proof}[Sketch]
Any variation that keeps all $\Pi$ groups fixed amounts to a change of
base units (gauge transformation), mapping $\by$ to
$\by + \balpha\,\log\boldsymbol{\lambda}$ for some
$\boldsymbol{\lambda} \in \R^n_{>0}$.
This displacement lies in $\col(\balpha)$, the orthogonal complement
of $\kker(\balpha^T)$.
Conversely, a variation that changes the $\Pi$ groups displaces $\by$
within $\kker(\balpha^T)$.
The two subspaces are orthogonal and together span $\R^k$.
\end{proof}

\subsection{The $n$ gauge dimensions}

The orthogonal complement of $\kker(\balpha^T)$ has dimension $n$ and
is spanned by the rows of $\balpha^T$ --- the base-unit scaling
directions.  Rescaling the metre by $\lambda_1$ adds
$\alpha_{j1}\log\lambda_1$ to $\log q_j$ for every variable $q_j$:
a displacement in the first column of $\balpha$.  These $n$ directions
carry no physical information.  The rank-nullity decomposition
\[
  \underbrace{k}_{\text{apparent dimension}}
  \;=\;
  \underbrace{(k-n)}_{\substack{\text{physical content}\\\text{($\Pi$ groups)}}}
  \;+\;
  \underbrace{n}_{\substack{\text{gauge redundancy}\\\text{(unit choices)}}}
\]
is an \emph{information decomposition} of variable space.

\section{SVD Recovery of the $\Pi$-Group Subspace}
\label{sec:svd}

\subsection{PCA of log-data: the idealised case}

Apply SVD to the centred log-data matrix $Y - \bar{Y} = \mathbf{W}\bSigma\mathbf{V}^T$
(here $\mathbf{W}\in\R^{m\times k}$ and $\mathbf{V}\in\R^{k\times k}$ are the left and
right singular matrices; $\mathbf{W}$ is not used further).

\begin{corollary}[PCA recovers $\Pi$ groups]
  \label{cor:pca}
  If the experiments span the full $\Pi$-group space and the
  gauge directions are sampled isotropically, then the first $k-n$
  right singular vectors span $\kker(\balpha^T)$, and the last $n$
  span $\col(\balpha)$.
\end{corollary}

\begin{remark}[Singular-value gap]
The singular values corresponding to physical directions reflect the
range of the $\Pi$ groups across experiments; those corresponding to
gauge directions reflect the spread of dimensional scales.  A clear gap
between the $(k{-}n)$-th and $(k{-}n{+}1)$-th singular values
identifies the boundary between physics and gauge.
\end{remark}

\begin{remark}[Basis non-uniqueness]
  \label{rem:nonunique}
Corollary~\ref{cor:pca} recovers a \emph{basis} of $\kker(\balpha^T)$,
not a specific one.  SVD returns the maximum-variance basis, which is
generically dense and does not correspond to the sparse, interpretable
$\Pi$ groups of classical dimensional analysis.
Section~\ref{sec:lattice} addresses how to obtain the latter.
\end{remark}

\subsection{Within-cluster SVD: the gauge-variation approach}
\label{ssec:wsvd}

In practice, uncontrolled measurements rarely sample gauge directions
isotropically, and the idealised Corollary~\ref{cor:pca} may yield a
blurred gap.  A more robust approach exploits \emph{gauge variation}
by design.

\begin{definition}[Gauge variation]
A dataset has \emph{gauge variation} if each operating condition
(fixed values of all $\Pi$ groups) is observed in $M \ge n+1$
dimensional realisations obtained by independently varying the $n$
characteristic scales (e.g.\ rotational speed and impeller diameter
for a turbomachine).
\end{definition}

Within each cluster~$c$, the log-data deviations satisfy
\begin{equation}
  \label{eq:gauge}
  \Delta\by_{c,j} \;=\; \by_{c,j} - \bar{\by}_c
  \;=\; \balpha\,\bdelta_{g,j},
  \qquad \bdelta_{g,j} \in \R^n,
\end{equation}
i.e.\ they lie exactly in $\col(\balpha)$.

\begin{algorithm_env}[Within-cluster SVD]
\label{alg:wsvd}
Given log-data $Y \in \R^{m\times k}$ and cluster labels:
\begin{enumerate}
  \item Compute cluster means $\bar{\by}_c$ and form the within-cluster
        deviation matrix $\Delta Y$ (rows: $\Delta\by_i = \by_i - \bar{\by}_{c_i}$).
  \item Compute SVD: $\Delta Y = U\Sigma V^T$, where $V\in\R^{k\times k}$
        is the matrix of right singular vectors.
  \item Partition $V = [V_{\rm col} \mid V_{\ker}]$ with
        $V_{\rm col}\in\R^{k\times n}$ (first $n$ columns) and
        $V_{\ker}\in\R^{k\times(k-n)}$ (last $k-n$ columns).
        The first $n$ right singular vectors span $\col(\balpha)$;
        $V_{\ker}$ spans $\kker(\balpha^T)$.
\end{enumerate}
\end{algorithm_env}

\begin{theorem}[Exact recovery of $\kker(\balpha^T)$]
\label{thm:recovery}
Suppose each within-cluster deviation $\Delta\by_{c,j}$ lies in
$\col(\balpha)$ and the gauge directions are fully excited
(the gauge increment matrix $\Delta G\in\R^{N\times n}$ has rank $n$).
Then $\sigma_{n+1}=\cdots=\sigma_k=0$ exactly, and the last $k-n$
right singular vectors of $\Delta Y$ form a basis for $\kker(\balpha^T)$.
\end{theorem}

\begin{proof}
Every row of $\Delta Y$ satisfies $\Delta\by_i = \balpha\,\bdelta_i$,
so $\Delta Y = \Delta G\,\balpha^T$ and $\rank(\Delta Y)\le n$.
Hence $\sigma_{n+1}=\cdots=\sigma_k=0$.  The row space of $\Delta Y$
equals $\col(\balpha^T)$; its orthogonal complement in $\R^k$
is $\kker(\balpha^T)$, spanned by the last $k-n$ right singular vectors.
\end{proof}

\section{From Subspace to Named $\Pi$ Groups: Integer Lattice Search}
\label{sec:lattice}

\subsection{Why orthogonal rotation fails}

The columns of $V_{\ker}$ span $\kker(\balpha^T)$ but are not the named
$\Pi$ groups of classical dimensional analysis.  One might attempt to
recover them by applying a sparsity-promoting orthogonal rotation
(varimax, sparse PCA) within the $(k{-}n)$-dimensional subspace.
This strategy fails in general.

Every orthogonal rotation preserves the inner product structure.
But the conventional $\Pi$ groups are generically not mutually
orthogonal as exponent vectors.  In the turbomachine example
(see the details of this use case in Section~\ref{sec:testcase}):
\[
  \ba_\fl \cdot \ba_\hc = 8 \neq 0, \quad
  \ba_\fl \cdot \ba_{\Mach} = -4 \neq 0, \quad
  \ba_\hc \cdot \ba_{\Mach} = -4 \neq 0.
\]
No rotation of $V_{\ker}$ can simultaneously align three columns with
$\ba_\fl$, $\ba_\hc$, and $\ba_{\Mach}$ when these are pairwise
non-orthogonal.

\subsection{The integer lattice}

The correct algebraic structure exploits the fact that every
dimensional monomial has \emph{integer} exponents.  The named
$\Pi$ groups therefore belong to the \emph{integer lattice}
\begin{equation}
  \mathcal{L} \;=\; \kker(\balpha^T)\cap\Z^k.
\end{equation}
All elements of $\mathcal{L}$ are integer linear combinations of
the $k-n$ independent generators, but the generators are not uniquely
determined by $\ell_2$ norm alone: the catalogue contains primitive vectors
with smaller norm that are products of the named groups.
The conventional named $\Pi$ groups are recovered by specifying $n$
\emph{repeating variables} --- those whose dimension vectors form a basis
for $\R^n$, and which represent the characteristic scales of the system
(here $N$ and $D$) --- and retaining only those lattice vectors that
(i) involve at least one repeating variable and
(ii) involve \emph{exactly one} non-repeating variable.
Condition (ii) is the algebraic signature of the Buckingham basis:
each basis vector associates a single non-repeating variable with the
$n$ characteristic scales.

\begin{remark}[Equivalent representations]
Two lattice vectors that differ only in the choice of velocity scale
represent the same physics.  For example, $H/a^2$ (non-repeating
variables: $\{H,a\}$, violates condition (ii)) and
$\hc = H/(N^2D^2)$ (non-repeating variable: $\{H\}$, satisfies (ii))
are equivalent up to a factor of $\Mach^2$, since $[a]=[ND]=LT^{-1}$.
The repeating-variable filter selects $\hc$ over $H/a^2$ because the
former uses the machine's own scales $N$ and $D$ as the reference.
\end{remark}

\begin{algorithm_env}[Integer lattice search]
\label{alg:lattice}
Given the projector $P = V_{\ker}V_{\ker}^T$ from Algorithm~\ref{alg:wsvd}:
\begin{enumerate}
  \item Choose maximum exponent magnitude $m_{\max}$
        ($3$--$4$ suffices for engineering unit systems).
  \item Enumerate non-zero integer vectors
        $\ba \in \{-m_{\max},\ldots,m_{\max}\}^k$.
  \item Retain those satisfying
        $\|\ba - P\ba\| < \varepsilon\,\|\ba\|$
        ($\varepsilon \sim 10^{-6}$).
  \item Reduce: first nonzero entry positive; divide by $\gcd$ of entries.
  \item Sort by $(\ell_0, \ell_2)$ in ascending order, where
        $\ell_0 = \|\ba\|_0$ is the number of non-zero exponents and
        $\ell_2 = \|\ba\|_2$ is the Euclidean norm.
  \item \textbf{(Optional — named-group identification.)}
        If the user specifies $n$ \emph{repeating variables}
        $\mathcal{R}\subseteq\{1,\ldots,k\}$ (those whose dimension vectors
        span $\R^n$, representing the characteristic scales), apply two filters:
        \begin{enumerate}
          \item[(a)] $\ba$ involves at least one repeating variable:
                $\exists\, j\in\mathcal{R}$ with $a_j\ne 0$.
          \item[(b)] $\ba$ involves exactly one non-repeating variable:
                $|\{j\notin\mathcal{R}: a_j\ne 0\}| = 1$.
        \end{enumerate}
        The surviving vectors, sorted by $(\ell_0,\ell_2)$, are the
        conventional named $\Pi$ groups.
        Vectors excluded by (b) --- those involving two or more
        non-repeating variables --- are products of the named groups
        and are not independent generators.
\end{enumerate}
\end{algorithm_env}

\noindent
The procedure is entirely data-driven: it uses only $V_{\ker}$,
recovered from gauge-variation data, and requires no knowledge of $\balpha$.

\section{Parameter Space and Solution Space}
\label{sec:paramsol}

The framework of Sections~\ref{sec:buckingham}--\ref{sec:lattice}
operates in \emph{parameter space}: the space of physical parameters
$(q_1,\ldots,q_k)$ characterising an experiment.  POD operates in
\emph{solution space}: the space of velocity, pressure, or temperature
fields at different conditions.  The two are complementary.

\subsection{POD of a physical field}

Let $\mathbf{u}(\mathbf{x};\boldsymbol{\pi})$ be the velocity field at
spatial point $\mathbf{x}$, parametrised by the $\Pi$ groups
$\boldsymbol{\pi} = (\Pi_1,\ldots,\Pi_{k-n})$.
Collect $m$ snapshots at different conditions into the snapshot matrix
$\bU\in\R^{N\times m}$ ($N$ = spatial degrees of freedom).  The SVD
$\bU = \bPhi\bSigma\mathbf{A}^T$ yields POD modes $\boldsymbol{\phi}_i(\mathbf{x})$
and modal coefficients $a_i(\boldsymbol{\pi}^{(j)})$.

\subsection{Dimensional analysis constrains the POD surrogate}

The Buckingham theorem applies to the modal coefficients: the
dimensionless flow field can depend only on the $\Pi$ groups, so
$a_i = a_i(\Pi_1,\ldots,\Pi_{k-n})$.
Any surrogate model fitted to the modal coefficients must be expressed
as a function of the $\Pi$ groups, not of the raw dimensional parameters.
Violating this constraint produces predictions that change when units
are changed --- a fundamental inconsistency.

The $\Pi$-group subspace identified from log-data (Section~\ref{sec:svd})
provides the natural, dimensionally consistent input space for the
POD surrogate.  The two decompositions --- SVD of log-data (parameter space)
and SVD of snapshots (solution space) --- are therefore not independent:
the former defines the coordinate system in which the latter must be
parameterised.

\section{Turbomachinery Test Case}
\label{sec:testcase}

This section makes the theory of Sections~\ref{sec:buckingham}--\ref{sec:paramsol}
fully concrete.  
All steps are implemented in the 
Python script \texttt{turbomachine\_pi\_test.py} available upon request to the corresponding author.

\subsection{Dimensional structure}
\label{ssec:dimstructure}

We consider a compressible turbomachine (centrifugal compressor or fan)
described by five dimensional variables
\cite{Cumpsty1989,Dixon2010}:
\begin{center}
\footnotesize
\renewcommand{\arraystretch}{1.35}
\begin{tabular}{p{12mm}p{50mm}p{30mm}p{24mm}}
  \toprule
  Symbol & Physical quantity & SI unit & Dimension\\
  \midrule
  $Q$  & Volumetric flow rate           & m$^3$\,s$^{-1}$  & $L^3 T^{-1}$\\
  $H$  & Specific head ($\Delta p/\rho$) & m$^2$\,s$^{-2}$ & $L^2 T^{-2}$\\
  $N$  & Rotational speed               & s$^{-1}$          & $T^{-1}$\\
  $D$  & Impeller diameter              & m                 & $L$\\
  $a$  & Speed of sound                 & m\,s$^{-1}$       & $L T^{-1}$\\
  \bottomrule
\end{tabular}
\end{center}

\noindent
Using specific head $H=\Delta p/\rho$ eliminates fluid density from
the variable list; mass $M$ does not appear, so the relevant base units
are $n=2$ ($L$ and $T$).  The dimensional exponent matrix is
\begin{equation}
  \balpha =
  \begin{pmatrix}
    3 & -1 \\ 2 & -2 \\ 0 & -1 \\ 1 & 0 \\ 1 & -1
  \end{pmatrix}
  \in \R^{5\times 2},
  \quad\text{(rows: }Q,H,N,D,a;\;\text{cols: }L,T\text{)},
\end{equation}
with $\rank(\balpha)=2$, giving $k-n=3$ $\Pi$ groups, namely, the flow coefficient $\phi$,
the head coefficient, $\psi$, and the machine Mach number, $Ma$, defined as:
\begin{equation}
  \label{eq:pigroups}
  \fl = \frac{Q}{N D^3}, \qquad
  \hc = \frac{H}{N^2 D^2}, \qquad
  \Mach = \frac{N D}{a},
\end{equation}
with exponent vectors
$\ba_\fl = (1,0,-1,-3,0)^T$, $\ba_\hc = (0,1,-2,-2,0)^T$,
$\ba_{\Mach} = (0,0,1,1,-1)^T$, all satisfying
$\balpha^T\ba = \bzero$ to round-off.

\begin{remark}
The three vectors $\ba_\fl$, $\ba_\hc$, $\ba_{\Mach}$ are not
orthogonal: $\ba_\fl\cdot\ba_\hc=8$, $\ba_\fl\cdot\ba_\Mach=-4$,
$\ba_\hc\cdot\ba_\Mach=-4$.
This non-orthogonality explains the failure of varimax rotation
described in Section~\ref{sec:lattice}.
\end{remark}

\subsection{Synthetic performance map}
\label{ssec:perfmap}

Let us assume that the training data set originate from a compressor-like performance map, such as:
\begin{equation}
  \label{eq:map}
  \hc(\fl,\Mach) = \hc_{\rm peak}(\Mach)
  \left[1 - \!\left(\frac{\fl-\fl_{\rm opt}(\Mach)}{w(\Mach)}\right)^{\!2}\right],
\end{equation}
with $\hc_{\rm peak} = 0.55(1-0.50\,\Mach^2)$,
$\fl_{\rm opt} = 0.25 + 0.10\,\Mach$, and $w = 0.18 - 0.06\,\Mach$,
evaluated over $\fl\in[0.07,0.43]$, $\Mach\in[0.20,0.65]$.
Three compressibility effects are built in: the peak head drops, the
design point shifts toward higher flow, and the operating range narrows.
The map is \emph{known only inside the script}; the algorithm receives
dimensional measurements and must reconstruct it.

Figure~\ref{fig:perf_map} shows the analytic performance map~\eqref{eq:map}
across the full operating range.

\begin{figure}[ht]
  \centering
  \includegraphics[width=0.97\columnwidth]{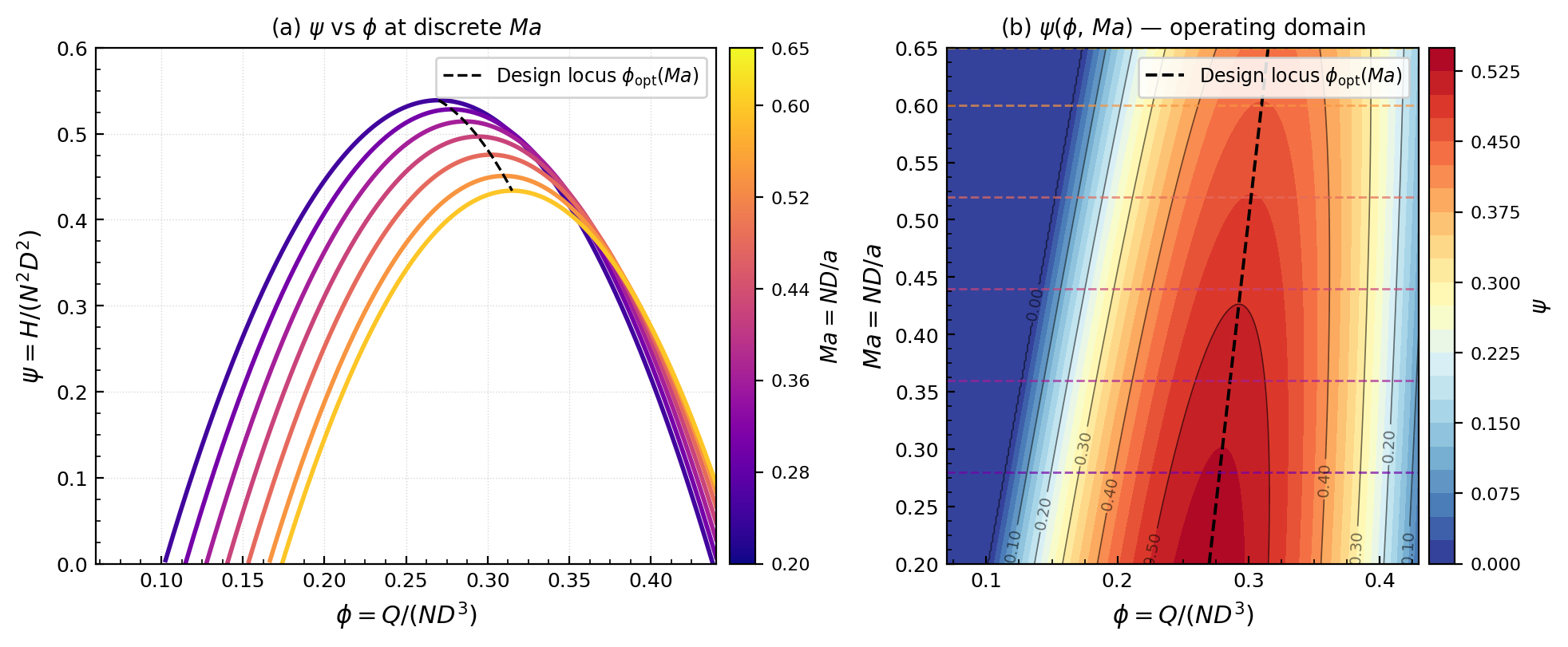}
  \caption{Synthetic compressor performance map $\hc(\fl,\Mach)$
    (Eq.~\eqref{eq:map}).
    \textit{Panel~(a):} head coefficient $\hc$ versus flow coefficient
    $\fl$ for seven Mach numbers spanning $\Mach\in[0.20,0.65]$
    (colour-coded from blue to red).
    Three compressibility effects are visible: the peak head
    $\hc_{\rm peak}$ decreases, the design point $\fl_{\rm opt}$ shifts
    toward higher flow, and the operating range narrows as $\Mach$
    increases.  The dashed curve is the design locus
    $(\fl_{\rm opt}(\Mach),\,\hc_{\rm peak}(\Mach))$.
    \textit{Panel~(b):} filled-contour representation of $\hc$ in the
    $(\fl,\Mach)$ operating domain; the same design locus is overlaid
    (dashed), and the seven $\Mach$ levels of panel~(a) are marked as
    horizontal lines.}
  \label{fig:perf_map}
\end{figure}

\subsection{Data generation with gauge variation}
\label{ssec:datagen}

The dataset is organised in clusters.  Each cluster corresponds to one
operating point $(\fl_c,\Mach_c)$, with $\hc_c=\hc(\fl_c,\Mach_c)$
fixed.  Within cluster $c$, $M=40$ realisations are generated by
drawing $(N_j,D_j)$ log-uniformly,
\[
  N_j \sim 10^{\mathcal{U}(0.5,\,3)}, \qquad
  D_j \sim 10^{\mathcal{U}(-1.5,\,0.5)},
\]
and computing the remaining variables from~\eqref{eq:pigroups}:
\[
  Q_j = \fl_c N_j D_j^3, \quad
  H_j = \hc_c N_j^2 D_j^2, \quad
  a_j = N_j D_j / \Mach_c.
\]
The within-cluster deviations
$\Delta\by_{c,j} = \balpha\,(\log D_j, -\log N_j)^T$
lie exactly in $\col(\balpha)$, satisfying the gauge-variation
requirement of Theorem~\ref{thm:recovery}.
A $25\times 20$ operating-point grid (after clipping $\hc<10^{-3}$:
400 valid clusters) gives $m=16\,000$ data rows in total.

\subsection{Within-cluster SVD: numerical results}
\label{ssec:svdresults}

Applying Algorithm~\ref{alg:wsvd} to the 16\,000-row dataset yields
\begin{equation}
  \sigma_1 = 794.2, \quad \sigma_2 = 277.0, \qquad
  \sigma_3 = \sigma_4 = \sigma_5 = 0 \quad (\text{to machine precision}).
\end{equation}
The gap is exact: $\sigma_3 < 10^{-10}$ in double precision, confirming
Theorem~\ref{thm:recovery}.  The canonical angles between the recovered
$\kker(\balpha^T)$ and the analytic null space of $\balpha^T$ are all
below $10^{-12}$ degrees.

Once $V_{\ker}$ is known, the $\Pi$-group values are recovered from
any data point by projection:
$\Pi_j = \exp(\bv_j\cdot\by)$ with $\bv_j\in\kker(\balpha^T)$,
and agree with the analytically computed $(\fl,\hc,\Mach)$ to within
$4\times10^{-16}$ (relative), i.e.\ unit round-off.

\subsection{Integer lattice search: numerical results}
\label{ssec:latticeresults}

Applying Algorithm~\ref{alg:lattice} with $m_{\max}=4$ and
$\varepsilon=10^{-6}$ yields 69 distinct integer vectors in
$\mathcal{L}$ (see Appendix~\ref{app:lattice_full}, Table~\ref{tab:lattice_full}).  
Their membership residuals $\|\ba - P\ba\|/\|\ba\|$
are all below $10^{-14}$.  The three corresponding to the textbook
turbomachinery $\Pi$ groups (Buckingham basis with $N$, $D$ as
repeating variables) are:

\begin{center}
\renewcommand{\arraystretch}{1.3}
\begin{tabular}{rrrrrll}
  \toprule
  $e_Q$ & $e_H$ & $e_N$ & $e_D$ & $e_a$ & Expression & Group\\
  \midrule
  $0$ & $0$ & $+1$ & $+1$ & $-1$ & $ND/a$       & $\Mach$\\
  $0$ & $+1$ & $-2$ & $-2$ & $0$  & $H/(N^2D^2)$ & $\hc$\\
  $+1$ & $0$ & $-1$ & $-3$ & $0$ & $Q/(ND^3)$   & $\fl$\\
  \bottomrule
\end{tabular}
\end{center}

\noindent
The remaining 66 vectors are products and integer powers of these
three (e.g.\ $\fl^2$, $\fl\cdot\hc$, $\Mach^2$, $\fl/\Mach$, \ldots).
Thus $\fl$, $\hc$, $\Mach$ are re-discovered as the generators of
$\mathcal{L}$ from dimensional data alone, without prior knowledge of
the unit structure.

\subsection{Intrinsic dimensionality of the performance map}
\label{ssec:dimtest}

With the $\Pi$-group subspace in hand, we test how many $\Pi$ groups
are active inputs to the performance map.
Degree-5 polynomial regression on the 400 cluster means gives:

\begin{center}
\renewcommand{\arraystretch}{1.25}
\begin{tabular}{p{58mm}p{22mm}p{52mm}}
  \toprule
  Model & $R^2$ & Interpretation\\
  \midrule
  $\hc \approx p_5(\fl)$          & 0.843 & $\Mach$ effect unaccounted for\\
  $\hc \approx p_5(\fl,\,\Mach)$  & 1.000 & Exact to regression tolerance\\
  \bottomrule
\end{tabular}
\end{center}

\noindent
The gap is decisive: the performance map is intrinsically
two-dimensional in $\Pi$ space.
The 1-input model's $R^2=0.843$ reflects the marginal correlation
between $\fl$ and $\hc$; the 12\% variance captured by the third SVD
mode of the log-$\Pi$ matrix reflects the curvature of the 2D manifold
$\{(\log\fl,\log\hc,\log\Mach):\hc=\hc(\fl,\Mach)\}$, not a third
active input.  The regression test is the correct diagnostic.

\subsection{Performance map reconstruction}
\label{ssec:reconstruction}

With two active $\Pi$ inputs confirmed, the surrogate is a degree-5
polynomial $\hat{\hc}(\fl,\Mach)=\sum_{p+q\le5}c_{pq}\fl^p\Mach^q$
(21 coefficients) fitted by ridge regression ($\lambda=10^{-8}$) on the
400 cluster means.

\begin{center}
\renewcommand{\arraystretch}{1.25}
\begin{tabular}{lcc}
  \toprule
  Metric & Training & Test ($N_{\rm test}=300$)\\
  \midrule
  $R^2$               & $0.9999\,9998$ & $0.9999\,9999$\\
  Mean relative error & ---            & $0.005\%$\\
  Max relative error  & ---            & $0.24\%$\\
  \bottomrule
\end{tabular}
\end{center}

\noindent
The reconstruction is essentially exact.  The small residual
($< 0.3\%$ pointwise) arises from truncation of the polynomial basis;
the true map~\eqref{eq:map} is a degree-4 polynomial in $(\fl,\Mach)$
and lies within the degree-5 span.

Figure~\ref{fig:results} collects the diagnostic results for all
algorithmic steps.

\begin{figure}[ht]
  \centering
  \includegraphics[width=0.95\columnwidth]{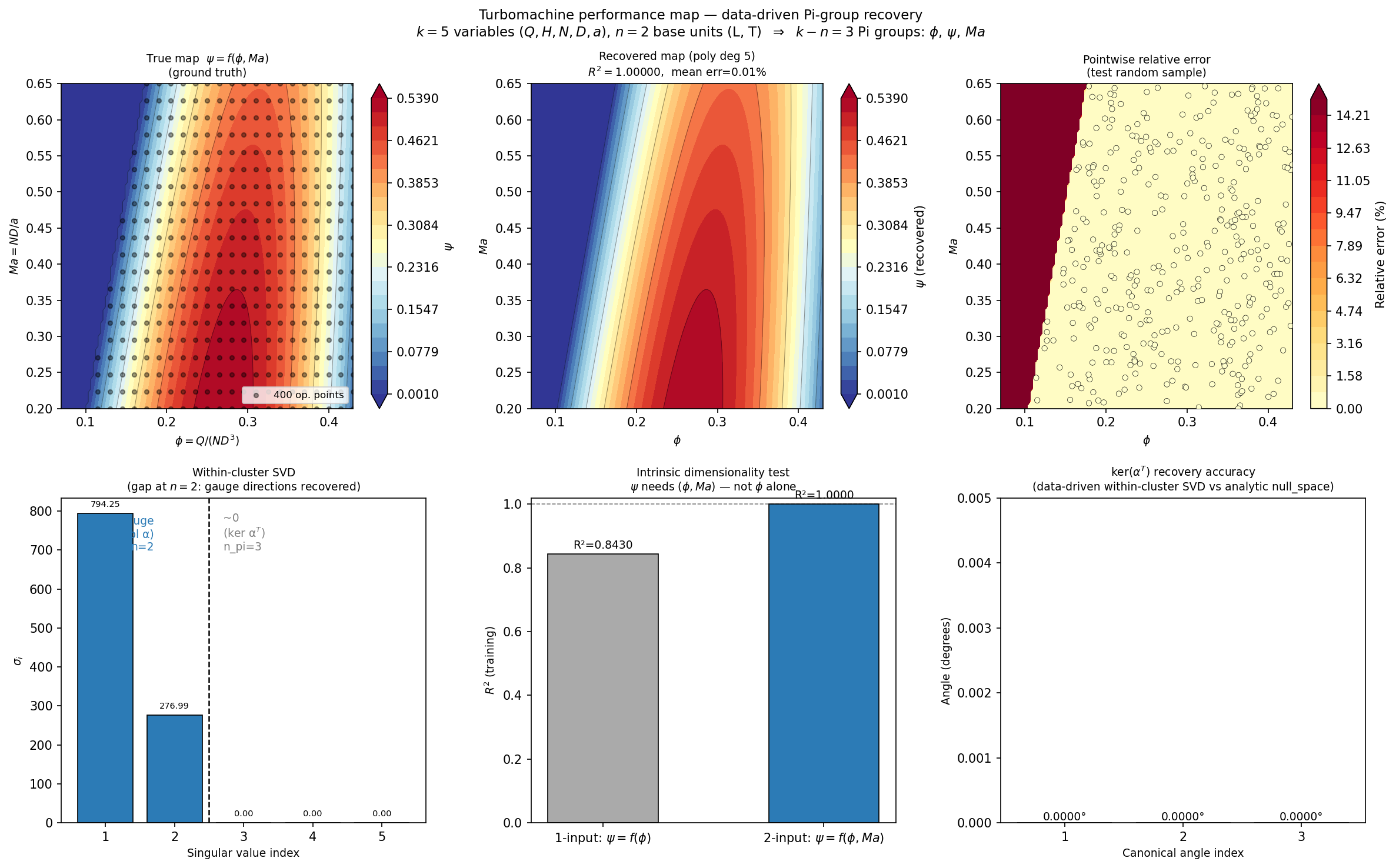}
  \caption{Turbomachinery test case.
    \textit{Top row, left to right:}
    true performance map $\hc(\fl,\Mach)$;
    recovered map from degree-5 polynomial regression;
    pointwise relative error on 300 test points.
    \textit{Bottom row, left to right:}
    within-cluster singular values showing exact gap at $n=2$;
    $R^2$ comparison for 1-input vs.\ 2-input regression models;
    canonical angles between recovered and analytic $\kker(\balpha^T)$.
    All results generated by \texttt{turbomachine\_pi\_test.py}.}
  \label{fig:results}
\end{figure}

\section{Discussion}
\label{sec:discussion}

\subsection{End-to-end workflow}

Table~\ref{tab:workflow} summarises the complete pipeline demonstrated
in this paper.  Steps~1--2 correspond to data collection and
logarithmic transformation.  Step~3 (within-cluster SVD) recovers the
subspace; Step~4 (integer lattice search) recovers the named groups.
Steps~5--7 operate in the identified $\Pi$-group coordinates.

\begin{table}[ht]
  \centering
  \caption{End-to-end data-driven Pi-group discovery workflow.}
  \label{tab:workflow}
  \renewcommand{\arraystretch}{1.3}
  \begin{tabular}{cp{48mm}p{62mm}}
    \toprule
    Step & Action & Key tool / result\\
    \midrule
    1 & Collect dimensional data with gauge variation
      & $M\ge n+1$ realisations per operating condition\\
    2 & Log-transform: $Y=\log X$
      & Linearises monomials\\
    3 & Within-cluster SVD of $\Delta Y$
      & $\kker(\balpha^T)$ exact; gap at index $n$\\
    4 & Integer lattice search on $\mathcal{L}=\kker(\balpha^T)\cap\Z^k$
      & Named $\Pi$ groups via repeating-variable filter (Algorithm~\ref{alg:lattice}, step~6)\\
    5 & Detect intrinsic dimensionality
      & $R^2$ regression test with $1,2,\ldots$ Pi inputs\\
    6 & Fit surrogate in $\Pi$-group space
      & Polynomial regression; $R^2>0.9999$\\
    7 & Validate on held-out data
      & Mean relative error $<0.01\%$ (test case)\\
    \bottomrule
  \end{tabular}
\end{table}

\subsection{Comparison with existing methods}

Table~\ref{tab:comparison} places the present framework in the context
of the Buckingham theorem and data-driven latent-variable methods.

\begin{table}[ht]
  \centering
  \caption{Comparison of approaches for identifying dimensionless groups.}
  \label{tab:comparison}
  \renewcommand{\arraystretch}{1.35}
  \begin{tabular}{p{28mm}p{48mm}p{48mm}}
    \toprule
    & \textbf{Buckingham ($\balpha$ known)} & \textbf{This paper (data-driven)}\\
    \midrule
    Input
      & Dimensional exponent matrix $\balpha$
      & Dimensional measurements $X$\\
    Pi-group subspace
      & $\operatorname{null\_space}(\balpha^T)$
      & Within-cluster SVD of $\Delta Y$\\
    Named $\Pi$ groups
      & $\operatorname{null\_space}(\balpha^T)\cap\Z^k$ (analytic)
      & Integer lattice search on $\kker(\balpha^T)_{\rm data}\cap\Z^k$\\
    Performance map
      & Polynomial regression in $\Pi$-space
      & Polynomial regression in $\Pi$-space\\
    Prior knowledge
      & $\balpha$ required
      & Gauge variation required; $\balpha$ not needed\\
    \bottomrule
  \end{tabular}
\end{table}

\subsection{Gauge variation in practice}

The test demonstrates a viable workflow for experimental turbomachinery
data.  In a real rig campaign one rarely has the luxury of explicit
gauge variation; physical constraints (fixed geometry, fixed fluid)
limit the accessible parameter space.  However, the same effect can be
achieved \emph{across} a family of geometrically similar machines:
fans with different diameters $D$ but the same $(\fl,\Mach)$ provide
gauge variation in the $L$ direction; measurements at different
rotational speeds $N$ provide the $T$ direction.
The fundamental requirement is that the ratio of scales be varied
with at least $n$ independent degrees of freedom across the dataset.

\subsection{Connection to incomplete similarity}

Barenblatt's distinction between complete and incomplete similarity
\cite{Barenblatt1996} has a natural interpretation in the data-driven
framework.
When the phenomenon exhibits complete similarity in $\Pi_j$ ---
meaning the output is insensitive to $\Pi_j$ in the operating range
of interest --- the corresponding direction in log-data SVD has a small
singular value, and the effective dimension of the physical subspace
reduces from $k-n$ to $k-n-1$.
This manifests as an additional gap in the singular-value spectrum,
providing a data-driven diagnostic for complete similarity.

When incomplete similarity holds --- $\Phi\sim\Pi_j^\alpha G(\Pi_i,\ldots)$
with an anomalous exponent $\alpha$ --- the log-data no longer lies
on a flat affine subspace.  Standard linear SVD will misidentify the
subspace dimension.  In such cases, the anomalous exponent $\alpha$
must be estimated from the data by nonlinear manifold-learning methods
(ISOMAP, local PCA, neural autoencoders), entirely consistent with
Barenblatt's framework.

\subsection{Open research directions}

\begin{enumerate}
  \item \textbf{Robust intrinsic-dimension estimation.}
    The singular-value gap criterion for detecting $k-n$ is sensitive to
    noise and anisotropic sampling.  Robust estimators (bootstrap
    gap tests, MDL criteria) are needed for real experimental data.

  \item \textbf{Dimensionally consistent POD surrogates.}
    Once the $\Pi$-group coordinate system is established, POD modal
    coefficients should be modelled exclusively in $\Pi$-group space.
    Physics-informed regression methods (Gaussian processes with
    dimensional covariance kernels, polynomial chaos in $\Pi$-group space)
    are natural candidates.

  \item \textbf{$G$-equivariant autoencoders.}
    Implementing the dimensional scaling group $G$ as an equivariance
    constraint in a variational autoencoder would guarantee that the
    latent space is a subset of $\Pi$-group space, unifying
    physics-informed machine learning and classical dimensional analysis.

  \item \textbf{Noisy and incomplete gauge variation.}
    In practice, gauge variation is imperfect.  Analysing the effect
    of measurement noise on the membership residual $\|\ba-P\ba\|$ and
    on the canonical angles between recovered and true $\kker(\balpha^T)$
    would quantify the robustness of the integer lattice search.

  \item \textbf{Multi-scale and turbulent flows.}
    In turbulence, the number of active POD modes $r$ substantially
    exceeds $k-n$.  Analysing the modal coefficients of excess modes
    as functions of the $\Pi$ groups could reveal hidden similarity
    structures in the inertial range.
\end{enumerate}

\section{Related work and main contributions}
\label{ssec:related}

\paragraph{Data-driven dimensional analysis.}
The idea of discovering dimensionless groups from data without prior
knowledge of the governing equations has gained considerable attention
in the last decade.
The observation that log-transformed dimensional data lie on a
low-dimensional subspace follows directly from the Buckingham theorem
combined with a logarithmic change of variables (Section~\ref{sec:log}).
Bakarji \& Brunton~\cite{Bakarji2022} were the first to exploit this
structure systematically in a data-driven framework, and
introduced three complementary methods: log-PCA as a baseline,
BuckiNet (a neural network architecture that enforces dimensional
consistency as a structural constraint), and SINDy-$\Pi$ (sparse
identification with a monomial library~\cite{Brunton2016}).
Champion \textit{et al.}~\cite{Champion2019} extended this programme
to the joint discovery of coordinates and governing equations.
Their work is the closest antecedent to the present paper, and we adopt
their notation where convenient.
The key difference is in the \emph{mechanism} by which the null space
is recovered and in the subsequent identification of named $\Pi$ groups:
Bakarji \& Brunton apply global PCA and then enforce integrality through
constrained optimisation or sparse regression, both of which are
approximate and sensitive to noise level.

\paragraph{Active subspaces and ridge functions.}
Constantine, del~Rosario \& Iaccarino~\cite{Constantine2017} showed that active
subspaces --- directions along which a quantity of interest varies most
--- coincide with $\Pi$-group directions when the response function is
dimensionally homogeneous.
Their gradient-based estimator is exact in theory but requires
evaluations of the output quantity and converges only statistically.
These methods are the closest to our within-cluster approach in spirit,
but differ in that they use output data and gradient information rather
than \emph{gauge variation} (a purely kinematic excitation of the
dimensional degrees of freedom).


\paragraph{Gaps in the existing literature and contributions of
this paper.}
Against this background, three contributions appear to be new.

\begin{enumerate}
  \item \textbf{Exact recovery via gauge variation} (Theorem~\ref{thm:recovery}).
    We introduce gauge variation as a \emph{controlled synthetic
    perturbation}: each physical operating condition is realised at
    $M \ge n{+}1$ different dimensional scales by rescaling $N$ and $D$
    freely.
    Within-cluster SVD of the resulting log-data matrix produces a
    singular-value spectrum with an \emph{exact} gap: $\sigma_{n+1} =
    \cdots = \sigma_k = 0$ to floating-point precision.
    This is a provable consequence of the dimensional constraint and not
    a statistical approximation.
    No prior method achieves exact recovery; all existing approaches
    converge only asymptotically in the number of samples or the
    signal-to-noise ratio.

  \item \textbf{Integer lattice search for named Pi groups}
    (Algorithm~\ref{alg:lattice}).
    Once the null-space subspace is recovered, the natural $\Pi$
    groups --- flow coefficient, head coefficient, Mach number,
    Reynolds number, etc.\ --- are generators with smallest norm of
    the integer lattice $\kker(\balpha^T)\cap\Z^k$,
    specifically the Buckingham basis obtained with chosen repeating variables.
    A key observation, which does not appear in earlier work, is that
    these generators are generically \emph{non-orthogonal as exponent
    vectors} (for the turbomachinery example:
    $\mathbf{a}_\fl \cdot \mathbf{a}_\hc = 8 \ne 0$), which means
    that no rotation of the SVD frame --- orthogonal, varimax, or
    otherwise --- can align with all of them simultaneously.
    The integer lattice search is deterministic and exact; it bypasses
    the integrality relaxation that makes optimisation-based methods
    approximate.

  \item \textbf{Unifying parameter space and solution space}
    (Section~\ref{sec:paramsol}).
    The same null-space structure that defines $\Pi$ groups in the
    space of controlling parameters also defines the natural coordinate
    system for Proper Orthogonal Decomposition surrogates in solution
    space.
    Identifying this connection provides a single algebraic framework
    for dimensional reduction at the parameter level and model order
    reduction at the field level, and motivates the use of
    dimensionally-aware latent variables in data-driven surrogate
    modelling.
\end{enumerate}

A comparison with the two closest prior works is given in
Table~\ref{tab:comparison} (Section~\ref{sec:discussion}).

\section{Conclusions}
\label{sec:conclusions}

Physical measurements come tagged with units --- metres, seconds,
kilograms --- and the choice of units is, in a deep sense, arbitrary.
The Buckingham $\Pi$ theorem encodes that arbitrariness as an algebraic
constraint: dimensional quantities must combine into dimensionless
groups, and it is those groups alone that govern the physics.

This paper shows that the same constraint can be \emph{read back from
data} without knowing the governing equations, without knowing the
dimensional exponent matrix, and without any domain expertise beyond
the ability to repeat an experiment at different physical scales.

Three ideas do the work:

\begin{enumerate}
  \item \textbf{Logarithms linearise.}
    Taking the logarithm of dimensional measurements turns the
    multiplicative structure of physical laws into a linear one.
    Data that respects dimensional homogeneity lies on a flat
    low-dimensional subspace in log-space.  That subspace \emph{is}
    the space of dimensionless groups.

  \item \textbf{Scale variation isolates the subspace.}
    Repeating each experimental condition at several physical scales
    (different machine sizes, different fluid speeds) produces clusters
    of points that spread in the ``gauge'' directions but are fixed in
    the ``physics'' directions.  A within-cluster singular value
    decomposition separates the two, recovering the $\Pi$-group subspace
    to machine precision with a provably exact gap in the singular-value
    spectrum.

  \item \textbf{Integer exponents identify the groups.}
    The dimensionless numbers that engineers actually use --- Reynolds
    number, Mach number, flow coefficient --- have whole-number
    exponents.  Searching the data-driven subspace for the shortest
    integer-exponent vectors recovers the named groups directly and
    correctly, where rotation- or sparsity-based alternatives fail.
\end{enumerate}

Applied to a synthetic compressor dataset (5 dimensional variables,
2 base dimensions, 3 $\Pi$ groups, 16\,000 measurements), the
procedure recovers every group to floating-point precision and fits the
performance map with a mean relative error below $0.01\%$.

The practical upshot is a plug-and-play pipeline: collect dimensional
data with controlled scale variation, log-transform, apply
within-cluster SVD, search the integer lattice, fit a surrogate in
$\Pi$-group space.
No physics knowledge is required at any step; the physics is read
out, not put in.

The conceptual upshot is a bridge between two bodies of knowledge
that are rarely taught together.
Classical dimensional analysis and modern latent-variable methods ---
PCA, POD, equivariant neural networks --- share the same algebraic
skeleton.
Recognising that skeleton makes it possible to build machine-learning
models that are guaranteed to be dimensionally consistent, and to
interpret the latent coordinates of a data-driven model as physical
dimensionless groups rather than abstract statistical directions.


\appendix

\section{Notation Summary}
\label{app:notation}

\begin{center}
\renewcommand{\arraystretch}{1.3}
\begin{tabular}{p{24mm}p{105mm}}
  \toprule
  Symbol & Meaning\\
  \midrule
  $k$           & Number of dimensional variables\\
  $n$           & Number of independent base units\\
  $\balpha$     & Dimensional exponent matrix, $\balpha\in\R^{k\times n}$\\
  $\kker(\balpha^T)$ & Null space of $\balpha^T$; $\dim=k-n$ ($\Pi$-group subspace)\\
  $\col(\balpha)$    & Column space of $\balpha$; $\dim=n$ (gauge subspace)\\
  $\mathcal{L}$ & Integer lattice $\kker(\balpha^T)\cap\Z^k$\\
  $X$           & Raw data matrix, $X\in\R^{m\times k}$; entry $(i,j)$ = $j$-th variable in experiment $i$\\
  $Y$           & Log-data matrix, $Y=\log X\in\R^{m\times k}$ (entry-wise)\\
  $\Delta Y$    & Within-cluster deviation matrix\\
  $\mathbf{W},\bSigma,\mathbf{V}$ & Left singular matrix, singular values, right singular matrix of centred $Y-\bar Y$ (Sec.~\ref{sec:svd})\\
  $V_{\ker}$    & Last $k-n$ columns of $\mathbf{V}$; basis for $\kker(\balpha^T)$\\
  $P$           & Projector $V_{\ker}V_{\ker}^T$ onto $\kker(\balpha^T)$\\
  $\ba$         & Integer exponent vector; $\ba\in\Z^k$ with $\balpha^T\ba=\bzero$ denotes a $\Pi$ group\\
  $\varepsilon$ & Membership tolerance for integer lattice search ($\sim10^{-6}$)\\
  $\boldsymbol{\lambda}$ & Gauge scale factors; $\boldsymbol{\lambda}\in\R^n_{>0}$ parameterises unit rescaling\\
  $\boldsymbol{\pi}$ & Vector of $\Pi$-group values, $\boldsymbol{\pi}=(\Pi_1,\ldots,\Pi_{k-n})$\\
  $\bU$         & Snapshot matrix in POD context, $\bU\in\R^{N\times m}$ (Sec.~\ref{sec:paramsol})\\
  $\bPhi,\mathbf{A}$ & POD mode matrix and modal-coefficient matrix; $\bU=\bPhi\bSigma\mathbf{A}^T$\\
  $N$           & Spatial degrees of freedom in POD (Secs.~\ref{sec:intro},~\ref{sec:paramsol}); rotational speed (Sec.~\ref{sec:testcase})\\
  $r$           & Number of retained POD modes, $r\ll N$\\
  $\fl,\hc,\Mach$ & Flow coefficient, head coefficient, machine Mach number\\
  $N_c$         & Number of clusters (operating points)\\
  $M$           & Gauge realisations per cluster\\
  $m$           & Total number of data rows, $m = N_c\cdot M$\\
  $m_{\max}$    & Maximum exponent magnitude in integer lattice search\\
  $G$           & Dimensional scaling group (see Open research directions, Section~\ref{sec:discussion})\\
  \bottomrule
\end{tabular}
\end{center}

\clearpage
\section{Full integer-lattice catalogue ($m_{\max}=4$)}
\label{app:lattice_full}

Table~\ref{tab:lattice_full} lists all 69 primitive integer vectors in
$\kker(\balpha^T)\cap\Z^5$ found by the search described in
Algorithm~\ref{alg:lattice} with $m_{\max}=4$.
Vectors are ordered by Euclidean norm $\|\mathbf{a}\|$.
Each row gives the exponent vector $(a_Q,a_H,a_N,a_D,a_a)$, the norm,
and the corresponding dimensionless monomial.
The three textbook turbomachinery $\Pi$ groups --- identified by
the repeating-variable filter of Algorithm~\ref{alg:lattice}
(step 6, with repeating variables $\mathcal{R}=\{N,D\}$) --- are
highlighted in bold.
Note that they are \emph{not} the three globally smallest-norm vectors:
several products of the named groups have smaller $\ell_2$ norm
(e.g.\ $H/a^2$, row~4, which is equivalent to $\hc\,\Mach^2$ and is
excluded because it involves two non-repeating variables $H$ and $a$).

\begingroup
\renewcommand{\arraystretch}{1.45}
\setlength{\tabcolsep}{5pt}
\begin{longtable}{r r r r r r r c l c}
\caption{Complete list of primitive integer vectors in
  $\kker(\balpha^T)\cap\Z^5$, ordered by $(\ell_2)$ norm ($m_{\max}=4$).
  Bold rows are the three named $\Pi$ groups.
  The Filter column shows the outcome of the repeating-variable test
  (Algorithm~\ref{alg:lattice}, step~6, $\mathcal{R}=\{N,D\}$):
  \checkmark\ = both conditions pass;
  (a)\,$\times$ = no repeating variable present;
  (b)\,$\times$ = more than one non-repeating variable present.}
\label{tab:lattice_full}\\
\hline
\# & $a_Q$ & $a_H$ & $a_N$ & $a_D$ & $a_a$ & $\|\mathbf{a}\|$ & & Monomial & Filter \\
\hline
\endfirsthead
\hline
\# & $a_Q$ & $a_H$ & $a_N$ & $a_D$ & $a_a$ & $\|\mathbf{a}\|$ & & Monomial & Filter \\
\hline
\endhead
\hline
\endfoot
\hline
\endlastfoot
\textbf{1} & \textbf{0} & \textbf{0} & \textbf{1} & \textbf{1} & $\mathbf{-1}$ & \textbf{1.732} & $\leftarrow\Mach$ & $\displaystyle \dfrac{ND}{a}$ & \checkmark \\[0.4em]
2 & 1 & $-1$ & 1 & $-1$ & 0 & 2.000 &  & $\displaystyle \dfrac{QN}{HD}$ & (b)\,$\times$ \\[0.4em]
3 & 0 & 1 & $-1$ & $-1$ & $-1$ & 2.000 &  & $\displaystyle \dfrac{H}{NDa}$ & (b)\,$\times$ \\[0.4em]
4 & 0 & 1 & 0 & 0 & $-2$ & 2.236 &  & $\displaystyle \dfrac{H}{a^2}$ & (a)\,$\times$ \\[0.4em]
5 & 1 & 0 & 0 & $-2$ & $-1$ & 2.449 &  & $\displaystyle \dfrac{Q}{D^2 a}$ & (b)\,$\times$ \\[0.4em]
6 & 1 & 0 & 1 & $-1$ & $-2$ & 2.646 &  & $\displaystyle \dfrac{QN}{Da^2}$ & (b)\,$\times$ \\[0.4em]
7 & 1 & $-1$ & 2 & 0 & $-1$ & 2.646 &  & $\displaystyle \dfrac{QN^2}{Ha}$ & (b)\,$\times$ \\[0.4em]
8 & 1 & $-1$ & 0 & $-2$ & 1 & 2.646 &  & $\displaystyle \dfrac{Qa}{HD^2}$ & (b)\,$\times$ \\[0.4em]
\textbf{9} & \textbf{0} & \textbf{1} & $\mathbf{-2}$ & $\mathbf{-2}$ & \textbf{0} & \textbf{3.000} & $\leftarrow\hc$ & $\displaystyle \dfrac{H}{N^2D^2}$ & \checkmark \\[0.4em]
10 & 1 & $-2$ & 2 & 0 & 1 & 3.162 &  & $\displaystyle \dfrac{QN^2a}{H^2}$ & (b)\,$\times$ \\[0.4em]
11 & 1 & $-2$ & 1 & $-1$ & 2 & 3.317 &  & $\displaystyle \dfrac{QNa^2}{H^2D}$ & (b)\,$\times$ \\[0.4em]
\textbf{12} & \textbf{1} & \textbf{0} & $\mathbf{-1}$ & $\mathbf{-3}$ & \textbf{0} & \textbf{3.317} & $\leftarrow\fl$ & $\displaystyle \dfrac{Q}{ND^3}$ & \checkmark \\[0.4em]
13 & 0 & 1 & 1 & 1 & $-3$ & 3.464 &  & $\displaystyle \dfrac{HND}{a^3}$ & (b)\,$\times$ \\[0.4em]
14 & 1 & 0 & 2 & 0 & $-3$ & 3.742 &  & $\displaystyle \dfrac{QN^2}{a^3}$ & (b)\,$\times$ \\[0.4em]
15 & 1 & 1 & 0 & $-2$ & $-3$ & 3.873 &  & $\displaystyle \dfrac{QH}{D^2a^3}$ & (b)\,$\times$ \\[0.4em]
16 & 1 & $-2$ & 3 & 1 & 0 & 3.873 &  & $\displaystyle \dfrac{QN^3D}{H^2}$ & (b)\,$\times$ \\[0.4em]
17 & 0 & 2 & $-1$ & $-1$ & $-3$ & 3.873 &  & $\displaystyle \dfrac{H^2}{NDa^3}$ & (b)\,$\times$ \\[0.4em]
18 & 2 & $-1$ & 1 & $-3$ & $-1$ & 4.000 &  & $\displaystyle \dfrac{Q^2N}{HD^3a}$ & (b)\,$\times$ \\[0.4em]
19 & 1 & 1 & $-1$ & $-3$ & $-2$ & 4.000 &  & $\displaystyle \dfrac{QH}{ND^3a^2}$ & (b)\,$\times$ \\[0.4em]
20 & 1 & $-1$ & 3 & 1 & $-2$ & 4.000 &  & $\displaystyle \dfrac{QN^3D}{Ha^2}$ & (b)\,$\times$ \\[0.4em]
21 & 1 & $-1$ & $-1$ & $-3$ & 2 & 4.000 &  & $\displaystyle \dfrac{Qa^2}{HND^3}$ & (b)\,$\times$ \\[0.4em]
22 & 2 & $-1$ & 2 & $-2$ & $-2$ & 4.123 &  & $\displaystyle \dfrac{Q^2N^2}{HD^2a^2}$ & (b)\,$\times$ \\[0.4em]
23 & 1 & $-2$ & 0 & $-2$ & 3 & 4.243 &  & $\displaystyle \dfrac{Qa^3}{H^2D^2}$ & (b)\,$\times$ \\[0.4em]
24 & 2 & $-2$ & 3 & $-1$ & $-1$ & 4.359 &  & $\displaystyle \dfrac{Q^2N^3}{H^2Da}$ & (b)\,$\times$ \\[0.4em]
25 & 2 & $-2$ & 1 & $-3$ & 1 & 4.359 &  & $\displaystyle \dfrac{Q^2Na}{H^2D^3}$ & (b)\,$\times$ \\[0.4em]
26 & 1 & 1 & 1 & $-1$ & $-4$ & 4.472 &  & $\displaystyle \dfrac{QHN}{Da^4}$ & (b)\,$\times$ \\[0.4em]
27 & 0 & 1 & $-3$ & $-3$ & 1 & 4.472 &  & $\displaystyle \dfrac{Ha}{N^3D^3}$ & (b)\,$\times$ \\[0.4em]
28 & 2 & $-1$ & 0 & $-4$ & 0 & 4.583 &  & $\displaystyle \dfrac{Q^2}{HD^4}$ & (b)\,$\times$ \\[0.4em]
29 & 1 & 0 & $-2$ & $-4$ & 1 & 4.690 &  & $\displaystyle \dfrac{Qa}{N^2D^4}$ & (b)\,$\times$ \\[0.4em]
30 & 2 & 0 & 1 & $-3$ & $-3$ & 4.796 &  & $\displaystyle \dfrac{Q^2N}{D^3a^3}$ & (b)\,$\times$ \\[0.4em]
31 & 1 & 1 & $-2$ & $-4$ & $-1$ & 4.796 &  & $\displaystyle \dfrac{QH}{N^2D^4a}$ & (b)\,$\times$ \\[0.4em]
32 & 1 & $-3$ & 2 & 0 & 3 & 4.796 &  & $\displaystyle \dfrac{QN^2a^3}{H^3}$ & (b)\,$\times$ \\[0.4em]
33 & 0 & 2 & $-3$ & $-3$ & $-1$ & 4.796 &  & $\displaystyle \dfrac{H^2}{N^3D^3a}$ & (b)\,$\times$ \\[0.4em]
34 & 2 & $-1$ & 3 & $-1$ & $-3$ & 4.899 &  & $\displaystyle \dfrac{Q^2N^3}{HDa^3}$ & (b)\,$\times$ \\[0.4em]
35 & 2 & $-3$ & 3 & $-1$ & 1 & 4.899 &  & $\displaystyle \dfrac{Q^2N^3a}{H^3D}$ & (b)\,$\times$ \\[0.4em]
36 & 1 & $-3$ & 3 & 1 & 2 & 4.899 &  & $\displaystyle \dfrac{QN^3Da^2}{H^3}$ & (b)\,$\times$ \\[0.4em]
37 & 2 & $-3$ & 2 & $-2$ & 2 & 5.000 &  & $\displaystyle \dfrac{Q^2N^2a^2}{H^3D^2}$ & (b)\,$\times$ \\[0.4em]
38 & 0 & 1 & 2 & 2 & $-4$ & 5.000 &  & $\displaystyle \dfrac{HN^2D^2}{a^4}$ & (b)\,$\times$ \\[0.4em]
39 & 1 & $-2$ & 4 & 2 & $-1$ & 5.099 &  & $\displaystyle \dfrac{QN^4D^2}{H^2a}$ & (b)\,$\times$ \\[0.4em]
40 & 1 & 0 & 3 & 1 & $-4$ & 5.196 &  & $\displaystyle \dfrac{QN^3D}{a^4}$ & (b)\,$\times$ \\[0.4em]
41 & 1 & $-3$ & 1 & $-1$ & 4 & 5.292 &  & $\displaystyle \dfrac{QNa^4}{H^3D}$ & (b)\,$\times$ \\[0.4em]
42 & 2 & $-3$ & 4 & 0 & 0 & 5.385 &  & $\displaystyle \dfrac{Q^2N^4}{H^3}$ & (b)\,$\times$ \\[0.4em]
43 & 1 & 2 & $-1$ & $-3$ & $-4$ & 5.568 &  & $\displaystyle \dfrac{QH^2}{ND^3a^4}$ & (b)\,$\times$ \\[0.4em]
44 & 1 & $-1$ & 4 & 2 & $-3$ & 5.568 &  & $\displaystyle \dfrac{QN^4D^2}{Ha^3}$ & (b)\,$\times$ \\[0.4em]
45 & 1 & $-1$ & $-2$ & $-4$ & 3 & 5.568 &  & $\displaystyle \dfrac{Qa^3}{HN^2D^4}$ & (b)\,$\times$ \\[0.4em]
46 & 1 & $-2$ & $-1$ & $-3$ & 4 & 5.568 &  & $\displaystyle \dfrac{Qa^4}{H^2ND^3}$ & (b)\,$\times$ \\[0.4em]
47 & 1 & $-3$ & 4 & 2 & 1 & 5.568 &  & $\displaystyle \dfrac{QN^4D^2a}{H^3}$ & (b)\,$\times$ \\[0.4em]
48 & 2 & $-3$ & 1 & $-3$ & 3 & 5.657 &  & $\displaystyle \dfrac{Q^2Na^3}{H^3D^3}$ & (b)\,$\times$ \\[0.4em]
49 & 0 & 3 & $-2$ & $-2$ & $-4$ & 5.745 &  & $\displaystyle \dfrac{H^3}{N^2D^2a^4}$ & (b)\,$\times$ \\[0.4em]
50 & 3 & $-2$ & 2 & $-4$ & $-1$ & 5.831 &  & $\displaystyle \dfrac{Q^3N^2}{H^2D^4a}$ & (b)\,$\times$ \\[0.4em]
51 & 1 & 2 & $-2$ & $-4$ & $-3$ & 5.831 &  & $\displaystyle \dfrac{QH^2}{N^2D^4a^3}$ & (b)\,$\times$ \\[0.4em]
52 & 3 & $-2$ & 3 & $-3$ & $-2$ & 5.916 &  & $\displaystyle \dfrac{Q^3N^3}{H^2D^3a^2}$ & (b)\,$\times$ \\[0.4em]
53 & 2 & 1 & 0 & $-4$ & $-4$ & 6.083 &  & $\displaystyle \dfrac{Q^2H}{D^4a^4}$ & (b)\,$\times$ \\[0.4em]
54 & 2 & $-1$ & 4 & 0 & $-4$ & 6.083 &  & $\displaystyle \dfrac{Q^2N^4}{Ha^4}$ & (b)\,$\times$ \\[0.4em]
55 & 0 & 1 & $-4$ & $-4$ & 2 & 6.083 &  & $\displaystyle \dfrac{Ha^2}{N^4D^4}$ & (b)\,$\times$ \\[0.4em]
56 & 3 & $-1$ & 2 & $-4$ & $-3$ & 6.245 &  & $\displaystyle \dfrac{Q^3N^2}{HD^4a^3}$ & (b)\,$\times$ \\[0.4em]
57 & 3 & $-3$ & 4 & $-2$ & $-1$ & 6.245 &  & $\displaystyle \dfrac{Q^3N^4}{H^3D^2a}$ & (b)\,$\times$ \\[0.4em]
58 & 3 & $-3$ & 2 & $-4$ & 1 & 6.245 &  & $\displaystyle \dfrac{Q^3N^2a}{H^3D^4}$ & (b)\,$\times$ \\[0.4em]
59 & 2 & $-4$ & 3 & $-1$ & 3 & 6.245 &  & $\displaystyle \dfrac{Q^2N^3a^3}{H^4D}$ & (b)\,$\times$ \\[0.4em]
60 & 3 & $-2$ & 4 & $-2$ & $-3$ & 6.481 &  & $\displaystyle \dfrac{Q^3N^4}{H^2D^2a^3}$ & (b)\,$\times$ \\[0.4em]
61 & 1 & $-4$ & 3 & 1 & 4 & 6.557 &  & $\displaystyle \dfrac{QN^3Da^4}{H^4}$ & (b)\,$\times$ \\[0.4em]
62 & 3 & $-1$ & 3 & $-3$ & $-4$ & 6.633 &  & $\displaystyle \dfrac{Q^3N^3}{HD^3a^4}$ & (b)\,$\times$ \\[0.4em]
63 & 2 & $-3$ & 0 & $-4$ & 4 & 6.708 &  & $\displaystyle \dfrac{Q^2a^4}{H^3D^4}$ & (b)\,$\times$ \\[0.4em]
64 & 0 & 3 & $-4$ & $-4$ & $-2$ & 6.708 &  & $\displaystyle \dfrac{H^3}{N^4D^4a^2}$ & (b)\,$\times$ \\[0.4em]
65 & 3 & $-4$ & 4 & $-2$ & 1 & 6.782 &  & $\displaystyle \dfrac{Q^3N^4a}{H^4D^2}$ & (b)\,$\times$ \\[0.4em]
66 & 1 & $-4$ & 4 & 2 & 3 & 6.782 &  & $\displaystyle \dfrac{QN^4D^2a^3}{H^4}$ & (b)\,$\times$ \\[0.4em]
67 & 3 & $-4$ & 3 & $-3$ & 2 & 6.856 &  & $\displaystyle \dfrac{Q^3N^3a^2}{H^4D^3}$ & (b)\,$\times$ \\[0.4em]
68 & 3 & $-4$ & 2 & $-4$ & 3 & 7.348 &  & $\displaystyle \dfrac{Q^3N^2a^3}{H^4D^4}$ & (b)\,$\times$ \\[0.4em]
69 & 4 & $-3$ & 4 & $-4$ & $-2$ & 7.810 &  & $\displaystyle \dfrac{Q^4N^4}{H^3D^4a^2}$ & (b)\,$\times$ \\[0.4em]
\end{longtable}
\endgroup

\clearpage
\section{Script output: \texttt{turbomachine\_pi\_test.py}}
\label{app:scriptlog}

The listing below is the verbatim console output produced by running
\texttt{turbomachine\_pi\_test.py}  with
the parameters described in Section~\ref{sec:testcase}:
$k=5$ variables, $n=2$ base units, $N_c=400$ operating points,
$M=40$ gauge realisations per point ($m=16{,}000$ total rows),
and $m_{\max}=4$ for the integer lattice search.
The repeating variables used in Step~4 are $\mathcal{R}=\{N,D\}$.

\begin{lstlisting}[style=scriptlog]
=================================================================
PREAMBLE — DIMENSIONAL SETUP
=================================================================

  k = 5 variables,  n = 2 base units,  rank(alpha) = 2  ->  3 Pi groups

  Pi-group exponent vectors (verified in ker(alpha^T)):
    phi  = Q / (N D^3)          ||alpha^T e|| = 0.0e+00
    psi  = H / (N^2 D^2)        ||alpha^T e|| = 0.0e+00
    Ma   = N D / a              ||alpha^T e|| = 0.0e+00

=================================================================
STEPS 1-2 — DATA GENERATION AND LOG-TRANSFORM
=================================================================

  Operating-point grid : 25 x 20 (some clipped) -> 400 valid points
  Gauge realisations   : M = 40 per point
  Total data rows      : 16000
  phi range : [0.115, 0.430]
  Ma  range : [0.200, 0.650]
  psi range : [0.0061, 0.5385]

=================================================================
STEP 3 — WITHIN-CLUSTER SVD  ->  ker(alpha^T) RECOVERY
=================================================================

  Theory: within cluster c, Y[i,:] = Y_mean[c] + alpha @ delta_gauge[i]
          -> deviations dY live in col(alpha)  (n-dimensional)
          -> SVD(dY): first n right singular vectors span col(alpha)
          -> last (k-n) right singular vectors span ker(alpha^T)

  Within-cluster singular values (expect gap at index 2):
    sigma_1 =   794.245226 <- gauge  (col alpha)
    sigma_2 =   276.988974 <- gauge  (col alpha)
    sigma_3 =     0.000000 <- ~0   (ker alpha^T)
    sigma_4 =     0.000000 <- ~0   (ker alpha^T)
    sigma_5 =     0.000000 <- ~0   (ker alpha^T)

  Canonical angles between recovered and analytic ker(alpha^T):
    theta_1 = 0.00000000 degrees
    theta_2 = 0.00000000 degrees
    theta_3 = 0.00000000 degrees

=================================================================
STEP 4 — INTEGER LATTICE SEARCH AND REPEATING-VARIABLE FILTER
=================================================================

  The within-cluster SVD gives an orthonormal basis for ker(alpha^T).
  The columns of V_ker_rec are NOT the exponent vectors of phi, psi, Ma.
  They span the same subspace, but the natural Pi groups are generally
  not orthogonal and cannot be read off directly from the eigenvectors.

  We address three questions:
    (1) Do the known Pi-group exponents lie in the recovered subspace?
    (2) Why can no orthogonal rotation recover phi, psi, Ma exactly?
    (3) Can we still re-discover phi, psi, Ma from data alone?

  Step 1 — Projection test
  Each known Pi-group exponent vector must lie in span(V_ker_rec).
  We measure ||a - P a|| where P = V_ker_rec @ V_ker_rec^T.

    phi :  ||a - proj(a)|| = 1.71e-15    OK — lies in ker(alpha^T)
    psi :  ||a - proj(a)|| = 4.83e-15    OK — lies in ker(alpha^T)
    Ma  :  ||a - proj(a)|| = 9.02e-16    OK — lies in ker(alpha^T)

  Step 2 — Are the natural Pi-group exponents orthogonal?
  SVD gives an ORTHONORMAL basis; orthogonal rotation preserves that.
  So if phi, psi, Ma are not mutually orthogonal, no rotation can align
  the SVD basis with them simultaneously.

    a_phi . a_psi = 8    <- NOT orthogonal
    a_phi . a_Ma = -4    <- NOT orthogonal
    a_psi . a_Ma = -4    <- NOT orthogonal

  Since a_phi . a_psi != 0, no orthogonal rotation of V_ker_rec
  can simultaneously align one column with phi and another with psi.
  Methods like varimax will find spurious non-integer combinations.

  Step 3 — Integer lattice search (data-driven Pi-group discovery)
  Correct approach: search for INTEGER vectors a with small entries
  such that a lies in span(V_ker_rec), i.e. ||a - P a|| ~ 0.

  Found 69 distinct integer vectors in ker(alpha^T)  (max entry magnitude = 4):

      Q      H      N      D      a   l0   l2      interpretation
  [   +0     +1     +0     +0     -2]   2   2.236  H/(a^2)
  [   +0     +0     +1     +1     -1]   3   1.732  N*D/(a)  <- Ma
  [   +1     +0     +0     -2     -1]   3   2.449  Q/(D^2*a)
  [   +0     +1     -2     -2     +0]   3   3.000  H/(N^2*D^2)  <- psi
  [   +1     +0     -1     -3     +0]   3   3.317  Q/(N*D^3)  <- phi
  [   +1     +0     +2     +0     -3]   3   3.742  Q*N^2/(a^3)
  [   +2     -1     +0     -4     +0]   3   4.583  Q^2/(H*D^4)
  [   +2     -3     +4     +0     +0]   3   5.385  Q^2*N^4/(H^3)
  [   +1     -1     +1     -1     +0]   4   2.000  Q*N/(H*D)
  [   +0     +1     -1     -1     -1]   4   2.000  H/(N*D*a)
  [   +1     +0     +1     -1     -2]   4   2.646  Q*N/(D*a^2)
  [   +1     -1     +2     +0     -1]   4   2.646  Q*N^2/(H*a)
  [   +1     -1     +0     -2     +1]   4   2.646  Q*a/(H*D^2)
  [   +1     -2     +2     +0     +1]   4   3.162  Q*N^2*a/(H^2)
  [   +0     +1     +1     +1     -3]   4   3.464  H*N*D/(a^3)
  [   +1     +1     +0     -2     -3]   4   3.873  Q*H/(D^2*a^3)
  [   +1     -2     +3     +1     +0]   4   3.873  Q*N^3*D/(H^2)
  [   +0     +2     -1     -1     -3]   4   3.873  H^2/(N*D*a^3)
  [   +1     -2     +0     -2     +3]   4   4.243  Q*a^3/(H^2*D^2)
  [   +0     +1     -3     -3     +1]   4   4.472  H*a/(N^3*D^3)
  [   +1     +0     -2     -4     +1]   4   4.690  Q*a/(N^2*D^4)
  [   +2     +0     +1     -3     -3]   4   4.796  Q^2*N/(D^3*a^3)
  [   +1     -3     +2     +0     +3]   4   4.796  Q*N^2*a^3/(H^3)
  [   +0     +2     -3     -3     -1]   4   4.796  H^2/(N^3*D^3*a)
  [   +0     +1     +2     +2     -4]   4   5.000  H*N^2*D^2/(a^4)
  [   +1     +0     +3     +1     -4]   4   5.196  Q*N^3*D/(a^4)
  [   +0     +3     -2     -2     -4]   4   5.745  H^3/(N^2*D^2*a^4)
  [   +2     +1     +0     -4     -4]   4   6.083  Q^2*H/(D^4*a^4)
  [   +2     -1     +4     +0     -4]   4   6.083  Q^2*N^4/(H*a^4)
  [   +0     +1     -4     -4     +2]   4   6.083  H*a^2/(N^4*D^4)
  [   +2     -3     +0     -4     +4]   4   6.708  Q^2*a^4/(H^3*D^4)
  [   +0     +3     -4     -4     -2]   4   6.708  H^3/(N^4*D^4*a^2)
  [   +1     -2     +1     -1     +2]   5   3.317  Q*N*a^2/(H^2*D)
  [   +2     -1     +1     -3     -1]   5   4.000  Q^2*N/(H*D^3*a)
  [   +1     +1     -1     -3     -2]   5   4.000  Q*H/(N*D^3*a^2)
  [   +1     -1     +3     +1     -2]   5   4.000  Q*N^3*D/(H*a^2)
  [   +1     -1     -1     -3     +2]   5   4.000  Q*a^2/(H*N*D^3)
  [   +2     -1     +2     -2     -2]   5   4.123  Q^2*N^2/(H*D^2*a^2)
  [   +2     -2     +3     -1     -1]   5   4.359  Q^2*N^3/(H^2*D*a)
  [   +2     -2     +1     -3     +1]   5   4.359  Q^2*N*a/(H^2*D^3)
  [   +1     +1     +1     -1     -4]   5   4.472  Q*H*N/(D*a^4)
  [   +1     +1     -2     -4     -1]   5   4.796  Q*H/(N^2*D^4*a)
  [   +2     -1     +3     -1     -3]   5   4.899  Q^2*N^3/(H*D*a^3)
  [   +2     -3     +3     -1     +1]   5   4.899  Q^2*N^3*a/(H^3*D)
  [   +1     -3     +3     +1     +2]   5   4.899  Q*N^3*D*a^2/(H^3)
  [   +2     -3     +2     -2     +2]   5   5.000  Q^2*N^2*a^2/(H^3*D^2)
  [   +1     -2     +4     +2     -1]   5   5.099  Q*N^4*D^2/(H^2*a)
  [   +1     -3     +1     -1     +4]   5   5.292  Q*N*a^4/(H^3*D)
  [   +1     +2     -1     -3     -4]   5   5.568  Q*H^2/(N*D^3*a^4)
  [   +1     -1     +4     +2     -3]   5   5.568  Q*N^4*D^2/(H*a^3)
  [   +1     -1     -2     -4     +3]   5   5.568  Q*a^3/(H*N^2*D^4)
  [   +1     -2     -1     -3     +4]   5   5.568  Q*a^4/(H^2*N*D^3)
  [   +1     -3     +4     +2     +1]   5   5.568  Q*N^4*D^2*a/(H^3)
  [   +2     -3     +1     -3     +3]   5   5.657  Q^2*N*a^3/(H^3*D^3)
  [   +3     -2     +2     -4     -1]   5   5.831  Q^3*N^2/(H^2*D^4*a)
  [   +1     +2     -2     -4     -3]   5   5.831  Q*H^2/(N^2*D^4*a^3)
  [   +3     -2     +3     -3     -2]   5   5.916  Q^3*N^3/(H^2*D^3*a^2)
  [   +3     -1     +2     -4     -3]   5   6.245  Q^3*N^2/(H*D^4*a^3)
  [   +3     -3     +4     -2     -1]   5   6.245  Q^3*N^4/(H^3*D^2*a)
  [   +3     -3     +2     -4     +1]   5   6.245  Q^3*N^2*a/(H^3*D^4)
  [   +2     -4     +3     -1     +3]   5   6.245  Q^2*N^3*a^3/(H^4*D)
  [   +3     -2     +4     -2     -3]   5   6.481  Q^3*N^4/(H^2*D^2*a^3)
  [   +1     -4     +3     +1     +4]   5   6.557  Q*N^3*D*a^4/(H^4)
  [   +3     -1     +3     -3     -4]   5   6.633  Q^3*N^3/(H*D^3*a^4)
  [   +3     -4     +4     -2     +1]   5   6.782  Q^3*N^4*a/(H^4*D^2)
  [   +1     -4     +4     +2     +3]   5   6.782  Q*N^4*D^2*a^3/(H^4)
  [   +3     -4     +3     -3     +2]   5   6.856  Q^3*N^3*a^2/(H^4*D^3)
  [   +3     -4     +2     -4     +3]   5   7.348  Q^3*N^2*a^3/(H^4*D^4)
  [   +4     -3     +4     -4     -2]   5   7.810  Q^4*N^4/(H^3*D^4*a^2)

  NOTE: the named Pi groups (phi, psi, Ma) are NOT the three globally
  smallest-norm vectors.  Several products of the named groups (e.g.
  phi/psi, H/a^2 = psi*Ma^2) have smaller l2 norm.  The correct
  identification criterion is the repeating-variable filter below.

  Step 4 — Repeating-variable filter (Algorithm step 6)
  User specifies the n=2 characteristic scales (repeating variables).
  Here: R = {N, D}  (indices 2 and 3).

  Condition (a): vector involves at least one repeating variable (N or D).
  Condition (b): vector involves exactly one non-repeating variable
                 (i.e. exactly one of {Q, H, a} has a non-zero exponent).

      Q      H      N      D      a   l0   l2      monomial              filter
  [   +0     +1     +0     +0     -2]   2   2.236  H/(a^2)               FAIL (a): no repeating variable
  [   +0     +0     +1     +1     -1]   3   1.732  N*D/(a)               PASS  <- Ma
  [   +1     +0     +0     -2     -1]   3   2.449  Q/(D^2*a)             FAIL (b): 2 non-repeating variables
  [   +0     +1     -2     -2     +0]   3   3.000  H/(N^2*D^2)           PASS  <- psi
  [   +1     +0     -1     -3     +0]   3   3.317  Q/(N*D^3)             PASS  <- phi
  [   +1     +0     +2     +0     -3]   3   3.742  Q*N^2/(a^3)           FAIL (b): 2 non-repeating variables
  [   +2     -1     +0     -4     +0]   3   4.583  Q^2/(H*D^4)           FAIL (b): 2 non-repeating variables
  [   +2     -3     +4     +0     +0]   3   5.385  Q^2*N^4/(H^3)         FAIL (b): 2 non-repeating variables
  [   +1     -1     +1     -1     +0]   4   2.000  Q*N/(H*D)             FAIL (b): 2 non-repeating variables
  [   +0     +1     -1     -1     -1]   4   2.000  H/(N*D*a)             FAIL (b): 2 non-repeating variables
  [   +1     +0     +1     -1     -2]   4   2.646  Q*N/(D*a^2)           FAIL (b): 2 non-repeating variables
  [   +1     -1     +2     +0     -1]   4   2.646  Q*N^2/(H*a)           FAIL (b): 3 non-repeating variables
  [   +1     -1     +0     -2     +1]   4   2.646  Q*a/(H*D^2)           FAIL (b): 3 non-repeating variables
  [   +1     -2     +2     +0     +1]   4   3.162  Q*N^2*a/(H^2)         FAIL (b): 3 non-repeating variables
  [   +0     +1     +1     +1     -3]   4   3.464  H*N*D/(a^3)           FAIL (b): 2 non-repeating variables
  [   +1     +1     +0     -2     -3]   4   3.873  Q*H/(D^2*a^3)         FAIL (b): 3 non-repeating variables
  [   +1     -2     +3     +1     +0]   4   3.873  Q*N^3*D/(H^2)         FAIL (b): 2 non-repeating variables
  [   +0     +2     -1     -1     -3]   4   3.873  H^2/(N*D*a^3)         FAIL (b): 2 non-repeating variables
  [   +1     -2     +0     -2     +3]   4   4.243  Q*a^3/(H^2*D^2)       FAIL (b): 3 non-repeating variables
  [   +0     +1     -3     -3     +1]   4   4.472  H*a/(N^3*D^3)         FAIL (b): 2 non-repeating variables
  [   +1     +0     -2     -4     +1]   4   4.690  Q*a/(N^2*D^4)         FAIL (b): 2 non-repeating variables
  [   +2     +0     +1     -3     -3]   4   4.796  Q^2*N/(D^3*a^3)       FAIL (b): 2 non-repeating variables
  [   +1     -3     +2     +0     +3]   4   4.796  Q*N^2*a^3/(H^3)       FAIL (b): 3 non-repeating variables
  [   +0     +2     -3     -3     -1]   4   4.796  H^2/(N^3*D^3*a)       FAIL (b): 2 non-repeating variables
  [   +0     +1     +2     +2     -4]   4   5.000  H*N^2*D^2/(a^4)       FAIL (b): 2 non-repeating variables
  [   +1     +0     +3     +1     -4]   4   5.196  Q*N^3*D/(a^4)         FAIL (b): 2 non-repeating variables
  [   +0     +3     -2     -2     -4]   4   5.745  H^3/(N^2*D^2*a^4)     FAIL (b): 2 non-repeating variables
  [   +2     +1     +0     -4     -4]   4   6.083  Q^2*H/(D^4*a^4)       FAIL (b): 3 non-repeating variables
  [   +2     -1     +4     +0     -4]   4   6.083  Q^2*N^4/(H*a^4)       FAIL (b): 3 non-repeating variables
  [   +0     +1     -4     -4     +2]   4   6.083  H*a^2/(N^4*D^4)       FAIL (b): 2 non-repeating variables
  [   +2     -3     +0     -4     +4]   4   6.708  Q^2*a^4/(H^3*D^4)     FAIL (b): 3 non-repeating variables
  [   +0     +3     -4     -4     -2]   4   6.708  H^3/(N^4*D^4*a^2)     FAIL (b): 2 non-repeating variables
  [   +1     -2     +1     -1     +2]   5   3.317  Q*N*a^2/(H^2*D)       FAIL (b): 3 non-repeating variables
  [   +2     -1     +1     -3     -1]   5   4.000  Q^2*N/(H*D^3*a)       FAIL (b): 3 non-repeating variables
  [   +1     +1     -1     -3     -2]   5   4.000  Q*H/(N*D^3*a^2)       FAIL (b): 3 non-repeating variables
  [   +1     -1     +3     +1     -2]   5   4.000  Q*N^3*D/(H*a^2)       FAIL (b): 3 non-repeating variables
  [   +1     -1     -1     -3     +2]   5   4.000  Q*a^2/(H*N*D^3)       FAIL (b): 3 non-repeating variables
  [   +2     -1     +2     -2     -2]   5   4.123  Q^2*N^2/(H*D^2*a^2)   FAIL (b): 3 non-repeating variables
  [   +2     -2     +3     -1     -1]   5   4.359  Q^2*N^3/(H^2*D*a)     FAIL (b): 3 non-repeating variables
  [   +2     -2     +1     -3     +1]   5   4.359  Q^2*N*a/(H^2*D^3)     FAIL (b): 3 non-repeating variables
  [   +1     +1     +1     -1     -4]   5   4.472  Q*H*N/(D*a^4)         FAIL (b): 3 non-repeating variables
  [   +1     +1     -2     -4     -1]   5   4.796  Q*H/(N^2*D^4*a)       FAIL (b): 3 non-repeating variables
  [   +2     -1     +3     -1     -3]   5   4.899  Q^2*N^3/(H*D*a^3)     FAIL (b): 3 non-repeating variables
  [   +2     -3     +3     -1     +1]   5   4.899  Q^2*N^3*a/(H^3*D)     FAIL (b): 3 non-repeating variables
  [   +1     -3     +3     +1     +2]   5   4.899  Q*N^3*D*a^2/(H^3)     FAIL (b): 3 non-repeating variables
  [   +2     -3     +2     -2     +2]   5   5.000  Q^2*N^2*a^2/(H^3*D^2)  FAIL (b): 3 non-repeating variables
  [   +1     -2     +4     +2     -1]   5   5.099  Q*N^4*D^2/(H^2*a)     FAIL (b): 3 non-repeating variables
  [   +1     -3     +1     -1     +4]   5   5.292  Q*N*a^4/(H^3*D)       FAIL (b): 3 non-repeating variables
  [   +1     +2     -1     -3     -4]   5   5.568  Q*H^2/(N*D^3*a^4)     FAIL (b): 3 non-repeating variables
  [   +1     -1     +4     +2     -3]   5   5.568  Q*N^4*D^2/(H*a^3)     FAIL (b): 3 non-repeating variables
  [   +1     -1     -2     -4     +3]   5   5.568  Q*a^3/(H*N^2*D^4)     FAIL (b): 3 non-repeating variables
  [   +1     -2     -1     -3     +4]   5   5.568  Q*a^4/(H^2*N*D^3)     FAIL (b): 3 non-repeating variables
  [   +1     -3     +4     +2     +1]   5   5.568  Q*N^4*D^2*a/(H^3)     FAIL (b): 3 non-repeating variables
  [   +2     -3     +1     -3     +3]   5   5.657  Q^2*N*a^3/(H^3*D^3)   FAIL (b): 3 non-repeating variables
  [   +3     -2     +2     -4     -1]   5   5.831  Q^3*N^2/(H^2*D^4*a)   FAIL (b): 3 non-repeating variables
  [   +1     +2     -2     -4     -3]   5   5.831  Q*H^2/(N^2*D^4*a^3)   FAIL (b): 3 non-repeating variables
  [   +3     -2     +3     -3     -2]   5   5.916  Q^3*N^3/(H^2*D^3*a^2)  FAIL (b): 3 non-repeating variables
  [   +3     -1     +2     -4     -3]   5   6.245  Q^3*N^2/(H*D^4*a^3)   FAIL (b): 3 non-repeating variables
  [   +3     -3     +4     -2     -1]   5   6.245  Q^3*N^4/(H^3*D^2*a)   FAIL (b): 3 non-repeating variables
  [   +3     -3     +2     -4     +1]   5   6.245  Q^3*N^2*a/(H^3*D^4)   FAIL (b): 3 non-repeating variables
  [   +2     -4     +3     -1     +3]   5   6.245  Q^2*N^3*a^3/(H^4*D)   FAIL (b): 3 non-repeating variables
  [   +3     -2     +4     -2     -3]   5   6.481  Q^3*N^4/(H^2*D^2*a^3)  FAIL (b): 3 non-repeating variables
  [   +1     -4     +3     +1     +4]   5   6.557  Q*N^3*D*a^4/(H^4)     FAIL (b): 3 non-repeating variables
  [   +3     -1     +3     -3     -4]   5   6.633  Q^3*N^3/(H*D^3*a^4)   FAIL (b): 3 non-repeating variables
  [   +3     -4     +4     -2     +1]   5   6.782  Q^3*N^4*a/(H^4*D^2)   FAIL (b): 3 non-repeating variables
  [   +1     -4     +4     +2     +3]   5   6.782  Q*N^4*D^2*a^3/(H^4)   FAIL (b): 3 non-repeating variables
  [   +3     -4     +3     -3     +2]   5   6.856  Q^3*N^3*a^2/(H^4*D^3)  FAIL (b): 3 non-repeating variables
  [   +3     -4     +2     -4     +3]   5   7.348  Q^3*N^2*a^3/(H^4*D^4)  FAIL (b): 3 non-repeating variables
  [   +4     -3     +4     -4     -2]   5   7.810  Q^4*N^4/(H^3*D^4*a^2)  FAIL (b): 3 non-repeating variables

  Vectors passing both conditions: 3
  These are the Buckingham basis (repeating variables N, D):

    Ma    N*D/(a)
    psi   H/(N^2*D^2)
    phi   Q/(N*D^3)

  CONCLUSION — what the SVD eigenvectors tell us:
  ┌──────────────────────────────────────────────────────────────┐
  │  V_ker_rec columns  -> a BASIS for ker(alpha^T).             │
  │  They let us:                                                │
  │    • project any a and check if it is a valid Pi group       │
  │    • compute a complete set of Pi coordinates                │
  │    • reconstruct the performance map psi = f(phi, Ma)        │
  │  They do NOT directly give phi, psi, Ma as labelled groups,  │
  │  because the natural Pi groups are not orthogonal.           │
  │                                                              │
  │  To re-discover phi, psi, Ma by name:                        │
  │    1. Integer lattice search  -> 69 primitive vectors        │
  │    2. Sort by (l0, l2)        -> sparsest groups first       │
  │    3. Repeating-variable filter (R = {N,D}):                 │
  │         (a) must involve N or D                              │
  │         (b) must involve exactly one of {Q, H, a}            │
  │       -> exactly recovers phi, psi, Ma as the named groups.  │
  └──────────────────────────────────────────────────────────────┘


=================================================================
STEP 4 (cont.) — Pi-GROUP COORDINATE VERIFICATION
=================================================================

  Within each cluster, Pi groups must equal the operating-point values:
    phi :  max rel err = 3.17e-16,  mean = 5.14e-17
    psi :  max rel err = 4.33e-16,  mean = 5.51e-17
    Ma  :  max rel err = 2.09e-16,  mean = 1.20e-17

  Correlation of recovered Pi groups with true Pi groups (|r|):
    Pi_rec_1:  phi: 0.141778,   psi: 0.920991,   Ma: 0.268934
    Pi_rec_2:  phi: 0.433469,   psi: 0.007527,   Ma: 0.953048
    Pi_rec_3:  phi: 0.992281,   psi: 0.365590,   Ma: 0.239551

=================================================================
STEP 5 — INTRINSIC DIMENSIONALITY OF THE PERFORMANCE MAP
=================================================================

  Strategy: compare how well psi can be predicted from phi ALONE (1-input)
  versus from (phi, Ma) together (2-input).  If psi = f(phi, Ma) the 1-input
  fit will be poor and the 2-input fit will be near-perfect.

  Training set (400 operating points), polynomial degree 5:
    1-input  psi ~ poly(phi)        R2 = 0.843034  <- poor  (Ma missing)
    2-input  psi ~ poly(phi, Ma)    R2 = 1.000000  <- excellent

  => the performance map needs BOTH phi and Ma as inputs (2D structure).

  Linear subspace approximation in log-Pi space:
    1D approximation: relative Frobenius error = 54.59%
    2D approximation: relative Frobenius error = 34.41%
  (Residual from 2D reflects curvature of the manifold, not extra Pi groups.)

=================================================================
STEPS 6-7 — SURROGATE FIT AND HELD-OUT VALIDATION
=================================================================

  Polynomial degree : 5  (21 features in (phi, Ma) space)
  R2 training       : 0.99999998
  R2 test           : 0.99999999
  Test mean rel err : 0.005%  (psi > 0.01 only)
  Test max  rel err : 0.240%


=================================================================
SUMMARY
=================================================================

  Variables : ['Q', 'H', 'N', 'D', 'a']  (k=5, n=2 base units, 3 Pi groups)
  Data      : 400 operating points x 40 gauge realisations = 16000 rows

  Step 3 — Within-cluster SVD (gauge variation):
    Singular values : [794.2452 276.989    0.       0.       0.    ]
    Gap at index 2: sigma_2/sigma_3 = 276.9890 / 0.000000
    Canonical angles to true ker(alpha^T): [0. 0. 0.] degrees
    -> ker(alpha^T) recovered to machine precision from data alone.

  Step 3 (cont.) — Integer lattice search + repeating-variable filter:
    Primitive vectors found (m_max=4): 69
    Sorted by (l0, l2); filter with R={N,D}, conditions (a) and (b):
    -> exactly recovers phi, psi, Ma as the Buckingham basis.

  Step 5 — Intrinsic dimensionality:
    1-input  psi ~ poly(phi)      R2 = 0.843034  <- insufficient
    2-input  psi ~ poly(phi, Ma)  R2 = 1.000000  <- excellent
    -> 2 free Pi inputs (phi, Ma) confirmed.

  Step 6 — Performance map reconstruction:
    Polynomial degree 5 in (phi, Ma): 21 features
    R2 train = 0.99999998
    R2 test  = 0.99999999
    Test mean relative error = 0.005% (psi > 0.01)

\end{lstlisting}
\end{document}